\documentclass[journal]{IEEEtran}

\usepackage{graphicx}
\usepackage{times} 
\usepackage{amsmath} 
\usepackage{amsthm}
\usepackage{amssymb}
\usepackage{bbm}
\usepackage{tikz}
\usetikzlibrary{positioning, arrows.meta, automata}
\usetikzlibrary{graphs}
\usepackage[colorlinks=false, urlcolor=blue, pdfborder={0 0 0}]{hyperref}
\usepackage{enumerate}
\usepackage{color}
\usepackage[linesnumbered,vlined,ruled]{algorithm2e}
\usepackage{subfigure}
\usepackage{comment}
\usepackage{cite}
\usepackage{tcolorbox}
\usepackage{verbatim} 
\usepackage{enumerate}
\usepackage{rotating}
\usepackage{enumitem}

\newtheorem{mydef}{Definition}
\newtheorem{mythm}{Theorem}
\newtheorem{myprob}{Problem}
\newtheorem{mylem}{Lemma}

\newtheorem{mypro}{Proposition}
\newtheorem{myexm}{Example}
\newtheorem{remark}{Remark}

\def \L{{\mathcal{L}}}

\def \S{{\mathfrak{S}}}
\def \A{{\mathfrak{A}}}

\title{Opacity Enforcing Supervisory Control\\ with a Priori Unknown Supervisors}

\author{
{Bohan Cui}, \IEEEmembership{Student Member, IEEE},  Ziyue Ma, \IEEEmembership{Senior Member, IEEE}, Alessandro Giua, \IEEEmembership{Fellow, IEEE}\\ and Xiang Yin, \IEEEmembership{Member, IEEE}

\thanks{This work was supported by the National Natural Science Foundation of
China (62173226, 62061136004, 92367203).}
\thanks{B. Cui and X. Yin are with the School of Automation and Sensing, Shanghai Jiao Tong University, Shanghai 200240, China, and also with the Key Laboratory of System Control and Information Processing, the Ministry of Education of China, Shanghai 200240, China. {\tt  E-mail: \{bohan\_cui, yinxiang\}@sjtu.edu.cn}. 
Ziyue Ma is with the School of Electro-Mechanical Engineering, Xidian University, Xi'an 710071, China. {\tt  E-mail: maziyue@xidian.edu.cn}.
Alessandro Giua is with the Department of Electrical and Electronic Engineering, University of Cagliari, Cagliari 09123, Italy. {\tt  E-mail: giua@unica.it}. (Corresponding Author: Xiang Yin)}
}

\begin{document}

\maketitle

\begin{abstract}
We investigate the enforcement of opacity in discrete-event systems via supervisory control. A system is said to be opaque if a passive intruder can never unambiguously infer whether the system is in a secret state through its observations.  In this context, the intruder’s knowledge about the supervisor plays a critical role in both problem formulation and solvability. Existing studies typically assume that the policy of the supervisor is either \emph{fully unknown} to the intruder or \emph{fully known a priori}, the latter leading to severe technical challenges and unresolved problems under incomparable observations.
This paper investigates opacity supervisory control under a new intermediate information setting, which we refer to as the \emph{a priori unknown} supervisor setting. In this setting, the supervisor’s internal realization is not publicly available, but the intruder can partially infer its behavior by eavesdropping on the control decisions issued online during system execution. We formalize the intruder’s information-flow under both observation-triggered and decision-triggered decision-issuance mechanisms and define the corresponding notions of opacity.
We provide sound and complete algorithms for synthesizing opacity-enforcing supervisors without imposing any restrictions on the observable or controllable event sets. 
By constructing an information-state structure that embeds the supervisor’s estimate of the intruder’s belief, the synthesis problem is reduced to a safety game. 
Finally, we show that, under strictly finer intruder observations, the proposed setting coincides with the standard a priori known supervisor model.
\end{abstract}


\IEEEpeerreviewmaketitle

\section{Introduction}

\subsection{Motivations} 

\IEEEPARstart{E}{nsuring} the privacy and security of system behaviors has become a critical challenge in modern engineering and cyber-physical systems (CPS), where information leakage or malicious observation can result in severe consequences. This paper investigates a fundamental security property, known as \emph{opacity} \cite{lafortune2018history,basilio2021analysis,liu2022secure}, within the framework of discrete-event systems (DES). In this setting, we assume the presence of a passive intruder that monitors the evolution of the system through its information-flow. A system is said to be opaque if the intruder can never infer with certainty that the system is in a secret state, thereby providing plausible deniability for sensitive operations.

As a fundamental system-theoretic property for information-flow security, opacity has attracted considerable attention in recent years from the  control systems community. 
In the context of DES, various notions of opacity have been proposed, including current/initial-state opacity \cite{lin2011opacity,saboori2013verification,han2023strong,miao2025always,basile2026sequence}, 
$K$/infinite-step opacity \cite{saboori2011verification,yin2017new,tong2022verification}, 
joint opacity \cite{zhu2022online,ritsuka2025joint} and pre-opacity \cite{yang2022secure}.
From the modeling perspective, beyond the finite-state automaton model of DES, opacity has also been investigated for infinite-state systems, such as Petri nets \cite{tong2017verification,berard2018complexity}, timed systems \cite{lefebvre2020exposure,gao2024state,dong2025k,zheng2025enforcement}, and continuous dynamical systems \cite{ramasubramanian2020notions,yin2021approximate,liu2021compositional,an2022enhancement}. Moreover, in addition to binary (opaque/non-opaque) analysis, several security-level metrics have been proposed to quantify the degree of opacity of a system  \cite{chen2016quantification,keroglou2018probabilistic,mu2022verifying,liu2024approximate}.

When the open-loop system is not opaque, an important problem is to enforce opacity via suitable enforcement mechanisms. In the DES literature, various opacity-enforcement mechanisms have been developed, such as supervisory control \cite{takai2008formula,ben2011supervisory,tong2018current,dubreil2010supervisory,saboori2011opacity,yin2015uniform,darondeau2015enforcing,xie2021opacity}, information-flow editing \cite{cassez2012synthesis,li2023opacity,liu2022opacity,udupa2025synthesis}, event encryption/delays \cite{falcone2015enforcement,lima2023ensuring,schonewille2024existence,reis2026enforcing}, and hybrid combinations of these techniques \cite{tai2023privacy,wintenberg2025integrating}. Among them, supervisory control is one of the most fundamental and widely investigated approaches.
The key idea is to synthesize a supervisor, in the Ramadge-Wonham framework \cite{wonham2018supervisory}, that dynamically disables/enables events online so as to eliminate secret-revealing behaviors, thereby rendering the resulting closed-loop system opaque.

\subsection{Status and Challenges in Opacity Supervisory Control}
In the context of supervisory control for opacity, one of the most important factors is \emph{what information the intruder knows about the closed-loop system under control}. This information directly affects the definition of opacity for controlled systems, since it determines which language the intruder uses to construct the system state estimate. Moreover, it also significantly influences the technical difficulty of the control synthesis problem. This is because it may be infeasible to decouple the intruder’s knowledge of the closed-loop behavior from the supervisor’s control policy, which is exactly what needs to be synthesized. Existing works in the literature can be broadly divided into the following two categories.

\textbf{Fully Unknown Supervisor Setting.}
A simple setting for opacity-enforcing supervisory control assumes that the intruder is \emph{fully unaware} of the supervisor’s implementation.
Therefore, the intruder still estimates the system state based on the open-loop language (i.e., without control), although some strings in this language are no longer feasible in the closed loop.
Under this assumption, the knowledge of the supervisor and the intruder can be fully decoupled.
Specifically, one can first construct the intruder’s state estimate based on the open-loop system, and then enforce opacity as a safety specification over this estimate from the supervisor’s perspective.
This setting has been completely solved in \cite{tong2018current} for finite-state systems,
where no restriction is imposed on the supervisor’s observable event set \(\Sigma_o\),
the supervisor’s controllable event set \(\Sigma_c\),
or the intruder’s observable event set \(\Sigma_a\).
A similar setting, in which the intruder is completely unaware of the supervisor, has also been considered in \cite{xie2021secure,shi2023synthesis,zhong2025secure}.

\textbf{A Priori Known Supervisor Setting.}
A more technically challenging setting assumes that the supervisor’s control policy is fully known \emph{a priori} to the intruder. 
That is, the control policy of the supervisor will become public information once it is designed offline.
Therefore, during the online execution phase, the intruder uses the \emph{closed-loop} language under supervision to estimate the system state, which is more precise than in the open-loop case.
This setting requires the control design to explicitly account for how control laws affect the intruder’s knowledge during the offline design stage.
Such a consideration is generally very difficult when the intruder’s and the supervisor’s observations are incomparable, in which case their information states cannot be decoupled.
To our knowledge, a general solution to the opacity control problem under this setting has remained open for almost 20 years.
Existing methods either rely on iterative designs without convergence guarantees \cite{takai2008formula,dubreil2008opacity}, or impose additional assumptions on the event sets \cite{dubreil2010supervisory,saboori2011opacity,yin2015uniform,xie2021opacity,moulton2022using}.

\subsection{Our Results and Contributions}
In this paper, we investigate the opacity-enforcing supervisory control problem from a new angle. Specifically, different from the fully unknown setting and the \emph{a priori} known setting studied in the literature, we consider a new setting in which the control policy of the synthesized supervisor is not directly public information to the intruder. Instead, the intruder can only access partial information about the supervisor by eavesdropping on the online control decisions that have already been issued. Therefore, we refer to this setting as the \emph{a priori unknown} setting. Note that the intruder’s knowledge about the supervisor along executed strings is still imperfect, since the supervisor and the intruder may have different observations. In particular, the intruder may observe that the supervisor has issued a new control decision, but does not know which supervisor observation triggered this decision update. 

The proposed \emph{a priori unknown} setting is of  importance  from both theoretical and practical perspectives. From the practical perspective, it captures realistic scenarios in which the supervisor’s control policy is not immediately available to the intruder, but is gradually revealed through its online interaction with the system. Such a situation arises in many applications; for example, the supervisor may need to 
broadcast or transmit its control decisions to actuators over communication channels that are not secure and may be eavesdropped on by an intruder. From the theoretical perspective, this setting identifies a new class of opacity-enforcing supervisory control problems, where the intruder can leverage more information than in the fully unknown setting, yet incomparable to the \emph{a priori} known setting. In particular, this new setting admits a complete solvability result without imposing any assumptions on the event sets; that is, the supervisor and the intruder may have incomparable observations.

\begin{figure}[tp]
	\centering
	\includegraphics[width=1\linewidth]{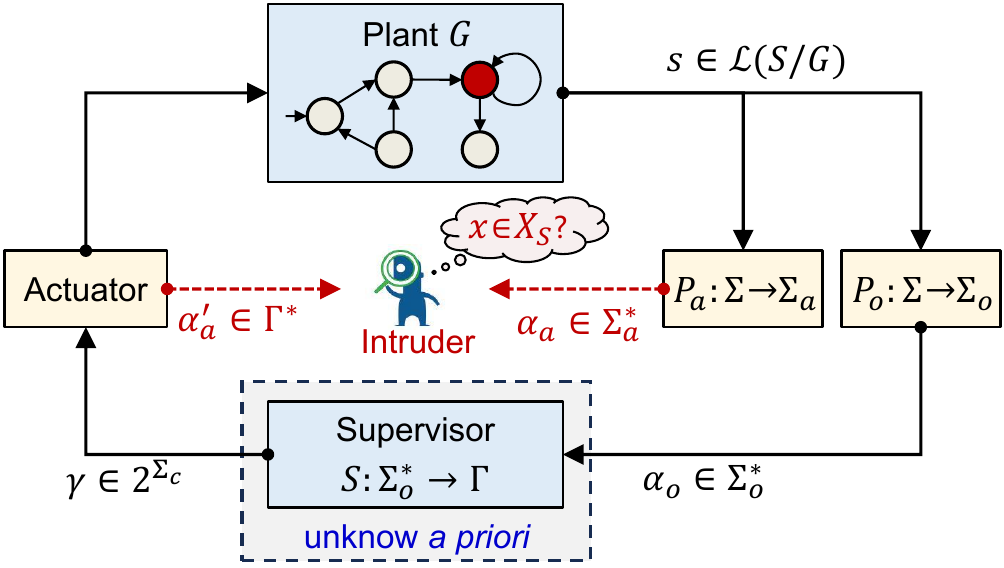}
	\caption{Conceptual illustration of  the \emph{a priori unknown} opacity supervisory control setting investigated.}
	\label{fig:illu}
\end{figure}

Our main contributions  are summarized as follows:
\begin{itemize}
    \item 
    We formulate a new class of opacity-enforcing supervisory control problems, where the internal realization of the supervisor is initially undisclosed, but the intruder can monitor the issuance of control decisions online during  system executions. In particular, we consider two distinct decision-issuance mechanisms: \emph{observation-triggered}, where the supervisor issues a decision upon observing a new event, and \emph{decision-triggered}, where a decision is issued only when the supervisor needs to modify the current control action applied. For both mechanisms, we define the intruder’s information-flow and the corresponding notion of opacity under the \emph{a priori} unknown supervisor setting.
    \item 
    We then develop an effective game-theoretic approach to solve the supervisor synthesis problem under this new setting. To capture the dynamic interaction among the plant, the supervisor, and the intruder, we construct a novel information structure that incorporates the supervisor’s estimate of the intruder’s knowledge. We characterize the intruder’s state-estimation process as a function of both its observations and the released control decisions. Based on this characterization, we reduce the synthesis problem to a safety game and provide systematic algorithms for computing an opacity-enforcing supervisor. We prove that our algorithms are sound and complete.
    \item 
    Our setting is more general than the fully unknown supervisor setting, while still enjoying complete solvability without imposing any assumptions on the event sets. That is, no restriction is imposed on \(\Sigma_o\),   \(\Sigma_c\), or  \(\Sigma_a\). 
    Compared with the \emph{a priori} known setting, the intruder's knowledge is incomparable since they leverage different information. On the other hand, this setting avoids the main technical difficulty caused by the coupling between the supervisor’s and the intruder’s information states, which leads to our decidability result.
    \item 
    Finally, under the assumption that the intruder’s observation is strictly finer than that of the supervisor, i.e., \(\Sigma_o \subseteq \Sigma_a\), we prove that  our new setting coincides with the standard \emph{a priori} known setting under this assumption. However, even in this case, the standard opacity control problem typically requires the additional assumption \(\Sigma_c \subseteq \Sigma_o\) \cite{dubreil2008opacity}, and the general case remains open to our knowledge. Therefore, our results also effectively solve the standard opacity-enforcement problem under \(\Sigma_o \subseteq \Sigma_a\) without further requiring that all controllable events be observable to the supervisor.
\end{itemize}

\subsection{Organization}
The rest of this paper is organized as follows. Section~\ref{section-pre} presents some necessary preliminaries. In Section~\ref{section-problem}, we formulate the opacity-enforcing control problem with \emph{a priori unknown} supervisors, focusing on the observation-triggered decision-issuance mechanism. Section~\ref{section-IS} investigates the evolution of information from both the intruder’s and the supervisor’s perspectives and introduces a novel information-state structure that captures both the supervisor’s knowledge and the intruder’s knowledge. In Section~\ref{section-alg}, we present a synthesis algorithm that first restricts the policy space to the IS-based control structure and then solves the resulting safety game; we also prove that the algorithm enforces opacity and is without loss of generality. In Section~\ref{section-change}, we further introduce the decision-triggered decision-issuance mechanism and solve the corresponding synthesis problem. Finally, Section~\ref{section-clu} concludes the paper.

\section{Preliminaries}\label{section-pre}

\subsection{System Model}

Let \(\Sigma\) be a finite set of events. A \emph{string} is a finite sequence of events, and we denote by \(\Sigma^*\) the set of all strings over \(\Sigma\) including the empty string \(\epsilon\). 
We denote $\Sigma^\epsilon=\Sigma\cup\{\epsilon\}$. 
For string \(s \in \Sigma^*\), the length of \(s\) is denoted by \(|s|\) with \(|\epsilon|=0\). A language \(L \subseteq \Sigma^*\) is a set of strings. For any string \(s\in L\), we denote by \(L/s\) the post-language of \(s\) in \(L\), i.e., \(L/ s:=\{w\in \Sigma^*: sw\in L\}\). The prefix-closure of \(L
\) is denote by \(\Bar{L}\), i.e., \(\Bar{L}=\{ u \in \Sigma^* : \exists v \in \Sigma^* \text{ s.t. } uv \in L \}\). A language \(L\subseteq\Sigma^*\) is said to be \emph{live} if \(\forall s\in L\), \(\exists\sigma\in\Sigma: s\sigma\in L\).

We consider a discrete-event system modeled by a deterministic finite-state automaton (DFA)
\[
G=(X,\Sigma,\delta,x_0),
\]
where \(X\) is the finite set of states, \(\Sigma\) is the finite set of events, \(\delta:X\times \Sigma \to X\) is the partial deterministic transition function such that \(\delta(x,\sigma)=x'\) means that there exists a transition from \(x\) to \(x'\) with event \(\sigma\), and \(x_0 \in X\) is the initial state.
The transition function \(\delta\) can also be extended to 
\(\delta: X\times \Sigma^*\rightarrow X\) 
recursively by: for any \(x\in X, s\in \Sigma^*, \sigma\in \Sigma\),
we have \(\delta(x, s\sigma) = \delta(\delta(x, s), \sigma)\) with \(\delta(x, \epsilon)=x\).
The language generated by \(G\) from state \(x\) is defined by 
\(\L(G,x)=\{s\in \Sigma^*:\delta(x,s)!\}\), where ``!" means ``is defined". 
The language generated by \(G\) is \(\L(G):=\L(G, x_0)\). For any \(s\in \L(G)\), we write \(\delta(x_0,s)\) simply as \(\delta(s)\). 
For any $x\in X$, we define $\Lambda(x)=\{\sigma\in \Sigma: \delta(x,\sigma)!\}$ as the set of active events at state $x$. For a set of states \(q\in 2^X\), we also write \(\bigcup_{x\in q}\Lambda(x)\) as \(\Lambda(q)\) for simplicity.

When the system is partially observed, \(\Sigma\) is partitioned as $\Sigma=\Sigma_o \dot{\cup} \Sigma_{uo}$,
where \(\Sigma_o\) is the set of observable events and \(\Sigma_{uo}\) is the set of unobservable events. 
The occurrence of each event is imperfectly observed through a natural projection \(P_o: \Sigma^* \to \Sigma_{o}^{*}\) defined as follows:
\begin{equation}
    P_o(\epsilon)=\epsilon \text{ and } \begin{aligned}
	P_o(s\sigma) = 
		\left\{
		\begin{array}{ll}
			 P_o(s)\sigma & \text{if}\quad    \sigma \in \Sigma_o  \\
			 P_o(s)                & \text{if}\quad    \sigma \in \Sigma_{uo} 
		\end{array}
		\right..   
\end{aligned} 
\end{equation}
The inverse projection \(P_o^{-1}:\Sigma^{*}_{o} \to 2^{\Sigma^*}\) is defined by \(P_o^{-1}(\alpha):=\{ s \in \Sigma^*: P_o(s)=\alpha \}\). For any observation \(\alpha\in P_o(\L(G))\), we define \(\mathcal{E}_{o}^{op}(\alpha)\) as the
current-state estimate that captures the set of all possible states
the (open-loop) system could be in currently when \(\alpha\) is observed under projection \(P_o\), i.e.,
\[
\mathcal{E}_{o}^{op}(\alpha)=\{ \delta(s) \in X: s \in P_o^{-1}(\alpha)\cap \L(G)\}.
\]

\subsection{Supervisory Control Theory}
In the supervisory control framework, a supervisor can restrict the behavior of the system \(G\) by dynamically disabling/enabling some system events. In this setting, the event set \(\Sigma\) is further partitioned as  $\Sigma= \Sigma_c \dot{\cup} \Sigma_{uc}$, 
where \(\Sigma_c\) is the set of controllable events and \(\Sigma_{uc}\) is the set of uncontrollable events. 
A control decision \(\gamma\in 2^\Sigma\) is said to be valid if \(\Sigma_{uc}\subseteq\gamma\), namely, uncontrollable events can never be disabled.
We define \(\Gamma=\{ \gamma \in 2^{\Sigma} : \Sigma_{uc} \subseteq \gamma \}\) as the set of valid control decisions. 
Since a supervisor can only make decisions based on its observations, a partial-observation supervisor is a function 
\[
S : P_o(\L(G)) \to \Gamma.
\]
We use the notation \(S/G\) to represent the controlled system and the language generated by \(S /G\), denoted by \(\L(S/G)\), is defined recursively in the following manners:
\begin{enumerate}
    \item \(\epsilon\in \L(S/G)\); and
    \item for any \(s\in \Sigma^*\), \(\sigma\in\Sigma\) we have \(s\sigma\in \L(S/G)\) iff \(s\sigma\in \L(G)\), \(s\in \L(S/G)\), and \(\sigma\in S(P_o(s))\).
\end{enumerate}
Since the supervisor knows its own control policy, 
then for any observation $\alpha\in P_o(\L/G)$, the 
\emph{controlled state estimate} w.r.t.\ $\Sigma_o$ and supervisor $S$ is 
\[
\mathcal{E}_{o}^S(\alpha)=\{ \delta(s) \in X: s \in P_o^{-1}(\alpha)\cap \L(S/G)\}.
\]

\subsection{Opacity under Passive Intruders}
We assume that system \(G\) has a ``secret", modeled as a set of secret states \(X_S\subseteq X\). 
We define the set of events whose occurrence can be observed by the intruder as \(\Sigma_a\subseteq\Sigma\). We also define \(\Sigma_{ua}=\Sigma\setminus\Sigma_a\) as the set of unobservable events from the intruder's perspective.
 The projection \(P_a:\Sigma^*\to\Sigma_a^*\), inverse projection \(P_a^{-1}:\Sigma_a^*\to 2^{\Sigma^*}\) and (open-loop)  state estimate \(\mathcal{E}_{a}^{op}:P_a(\L(G))\to 2^X\) are also defined in the same way as \(P_o\), \(P_o^{-1}\) and \(\mathcal{E}_{o}^{op}\), respectively. 
In general, \(\Sigma_a\) and \(\Sigma_o\) can be incomparable, i.e., both \(\Sigma_o\not\subseteq \Sigma_a\) and \(\Sigma_a\not\subseteq \Sigma_o\) may hold.

In opacity analysis, we consider a passive \emph{intruder} that has the following capabilities:
\begin{itemize}
    \item [\textbf{A1}:] 
    it knows the system model $G$; 
    \item [\textbf{A2}:] 
    it can eavesdrop on the occurrences of  events in $\Sigma_a$.
\end{itemize} 
A (open-loop) system is said to be current-state opaque if the intruder can never determine for sure that the system is currently in a secret state based on its observation. 
\begin{mydef}[\bf  Current-State Opacity]\upshape
    Given system \(G\), set of observable events \(\Sigma_a\) of the intruder, and set of secret states \(X_S\), we say \(G\) is current-state opaque (w.r.t. \(\Sigma_a\) and \(X_S\)) if \vspace{-3pt}
    \begin{align}
        (\forall s\in \L(G))[\mathcal{E}_{a}^{op}(P_a(s))\not\subseteq X_S].
    \end{align}
\end{mydef}\vspace{-3pt}

We finally define some operators that will be used later. 
The \emph{unobservable reach} of the supervisor for  state set \(q \subseteq X\) under control decision \(\gamma \in \Gamma\) is   \vspace{-3pt}
\[
\text{UR}_{\gamma,o}(q)=\{ \delta(x,w) \in X: x \in q, w \in (\Sigma_{uo}\cap \gamma)^* \}. \vspace{-3pt}
\]
We also define  the  unobservable reach of the intruder for  states \(q \subseteq X\) under
  \(\gamma \in \Gamma\)  w.r.t. $\Sigma_{ua}$ by \(\text{UR}_{\gamma,a}(q)\).
Moreover, we define 
\[
\text{UR}^{+}_{\gamma,a}(q)=\{ \delta(x,w) \!\in\! X: x \!\in\! q, w \!\in\! (\Sigma_{ua}\cap \gamma)^*\setminus \{\epsilon\} \} 
\]
as the intruder's unobservable reach in one or more steps, starting from the subset of states \(q \subseteq X\) under
the control decision \(\gamma \in \Gamma\).

The \emph{observable reach} of   state set \(q \subseteq X\) under the event \(\sigma \in \Sigma_o\cup\Sigma_a\) is given by \vspace{-3pt}
\[
\text{NX}_\sigma(q)=\{ \delta(x,\sigma)\in X: x \in q  \}.\vspace{-3pt}
\]

\section{Opacity with A Priori Unknown Supervisors}\label{section-problem}

When the open-loop system is not opaque, our objective is to design a supervisor such that the resulting closed-loop system is opaque. As discussed in the introduction, we consider the \emph{a priori unknown supervisor} setting, where the supervisor reveals to the intruder only the control decisions issued online, rather than its complete control policy \emph{a priori}.
More specifically, in addition to Assumptions A1 and A2, this setting is further characterized by the following two assumptions:
\begin{itemize}
\item[\textbf{A3}:]
The intruder does not know the supervisor’s control policy \emph{a priori}, but it can observe every control decision issued by the supervisor as the system evolves.
\item[\textbf{A4}:]
The supervisor issues a control decision \emph{immediately} upon the occurrence of each newly observed event.
\end{itemize}
Note that \textbf{A4} corresponds to the \emph{observation-triggered} decision-issuance mechanism, since the supervisor updates (and issues) a control decision after each observation, even when two successive decisions are identical. In some applications, to reduce bandwidth consumption, the supervisor may transmit a control decision to the actuator only when it needs to change the current decision applied. Such a decision-issuance mechanism is referred to as \emph{decision-triggered}, and it will be discussed in Section~\ref{section-change}. 
In this paper, we assume that the intruder is unaware of whether the supervisor employs an observation-triggered or a decision-triggered decision-issuance mechanism.  

Formally, based on the above setting, given  a supervisor \(S: P_o(\L(G))\to \Gamma\), 
the \emph{information-flow} of the controlled system from the intruder's perspective is captured as the function:
\[
\mathbb{P}_S: \L(S/G)\to (\Sigma_a^\epsilon  \times \Gamma^\epsilon)^* 
\]
such that 
\begin{itemize}
    \item[(i)]
     $ \mathbb{P}_S(\epsilon)=(\epsilon, S(\epsilon) )$;
and 
    \item[(ii)]
    for any string $s\sigma \in \L(S/G)$, we have
\begin{equation}
 \begin{aligned}
	\mathbb{P}_S(s\sigma) \!=\! 
		\left\{
		\begin{array}{ll}
			 \mathbb{P}_S(s)(\sigma ,\epsilon)            & \text{if }     \sigma \in \Sigma_a\setminus \Sigma_o  \\
			 \mathbb{P}_S(s) (\epsilon, S(P_o(s\sigma)))  & \text{if }     \sigma \in \Sigma_o\setminus \Sigma_a   \\
          \mathbb{P}_S(s) (\sigma, S(P_o(s\sigma)))    & \text{if }     \sigma \in \Sigma_o\cap \Sigma_a \\
          \mathbb{P}_S(s)  & \text{if }     \sigma \in \Sigma_{uo}\cap  \Sigma_{ua} 
		\end{array}
		\right..\!\!   
\end{aligned} 
\end{equation}
\end{itemize} 
We denote by $\mathcal{O}=\Sigma_a^\epsilon  \times \Gamma^\epsilon$ the observation set of the intruder. 
Intuitively, the information-flow function $\mathbb{P}_S$ maps each string $s$ to a sequence of observation pairs. In each pair $(e, \gamma)$, the first element $e \in \Sigma_a \cup \{\epsilon\}$ represents an observed event (or an empty observation if the event is unobservable), and the second element $\gamma \in \Gamma \cup \{\epsilon\}$ represents the control pattern issued by the supervisor (or an empty observation if no control update occurs). The above definition captures the following four cases, depending on event observability from the perspectives of the supervisor and the intruder:
\begin{itemize}
    \item 
     When \(\sigma \in \Sigma_a \setminus \Sigma_o\), only the intruder can observe the newly occurred event. Therefore, the event \(\sigma\) is appended to the information-flow, but no new control decision is issued by the supervisor as it observes nothing.
     \item 
     When \(\sigma \in \Sigma_o \setminus \Sigma_a\), only the supervisor can observe the newly occurred event. Therefore, the supervisor updates its control decision to \(S(P_o(s\sigma))\), which is available to the intruder. However, the intruder does not know which specific observation triggered this decision update.
     \item 
     When \(\sigma \in \Sigma_o \cap \Sigma_a\), both the supervisor and the intruder can observe the event. Therefore, the tuple \(  (   \sigma,\, S(P_o(s\sigma))   )     \) is appended to the information-flow.
     \item 
     When \(\sigma \in \Sigma_{uo} \cap \Sigma_{ua}\), neither the supervisor nor the intruder is aware of the event occurrence; hence, the information-flow remains unchanged.
\end{itemize}

We use the following example to illustrate the information-flow of the intruder with a priori unknown supervisors. 

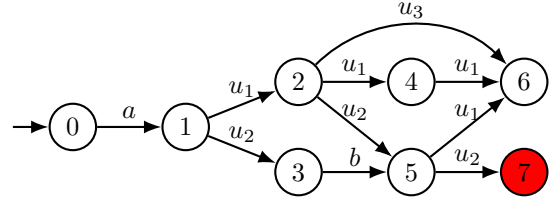
\begin{figure}
    \centering
       \begin{tikzpicture}[->,>={Latex}, thick, initial text={}, node distance=1.5cm, initial where=left, thick, base node/.style={circle, draw, minimum size=6mm}]  
   \node[state, initial, base node, ] (0) {$0$};
   \node[state, base node, ] (1) [right of=0] {$1$};
   \node[state, base node, ] (2) [xshift=3cm,yshift=0.6cm] {$2$};
   \node[state, base node, ] (3) [xshift=3cm,yshift=-0.6cm] {$3$};
   \node[state, base node, ] (4) [right of=2] {$4$};
   \node[state, base node, ] (5) [right of=3] {$5$};
   \node[state, base node, ] (6) [right of=4] {$6$};
   \node[state, base node, fill=red] (7) [right of=5] {$7$};
   
   \path[->]
   (0) edge node [yshift=0.2cm] {$a$} (1)
   (2) edge node [yshift=0.2cm] {$u_1$} (4)
   (3) edge node [yshift=0.2cm] {$b$} (5)
   (1) edge node [yshift=0.2cm] {$u_1$} (2)
   (1) edge node [yshift=0.2cm] {$u_2$} (3)
   (2) edge node [yshift=0.2cm] {$u_2$} (5)
   (5) edge node [yshift=0.2cm] {$u_2$} (7)
   (4) edge node [yshift=0.2cm] {$u_1$} (6)
   (5) edge node [yshift=0.2cm] {$u_1$} (6);

   \draw[->]
   (2) to [bend left=45] node [yshift=0.2cm] {$u_3$} (6);

   \end{tikzpicture}
    \caption{System \(G\), where \(\Sigma_o=\{u_1,u_2\}\), \(\Sigma_a=\{a,b\}\), \(\Sigma_c=\{u_1,u_2,u_3\}\).}
    \label{fig:motiexample}
\end{figure}

\begin{myexm}[Online Information-Flow]\label{example1}
 Let us consider  system \(G\)  shown in Figure~\ref{fig:motiexample}, where \(\Sigma_o=\{u_1,u_2\}\), \(\Sigma_a=\{a,b\}\), \(\Sigma_c=\{u_1,u_2,u_3\}\). Suppose \(G\) is controlled by a supervisor \(S:P_o(\L(G))\to \Gamma\) defined by: \(S(u_1)=S(u_1u_2)=\{a,b,u_2\}\), \(S(u_2)=S(u_2u_1)=\{a,b,u_1\}\), and \(S(\alpha)=\Sigma\) for any other \(\alpha\in P_o(\L(G))\). 
 For string \(a u_1 u_2 u_2\in\L(S/G)\), 
 the information-flow of the intruder is 
\[\mathbb{P}_S(a u_1 u_2 u_2)=
(\epsilon,\Sigma) (a,\epsilon)(\epsilon, \{a,b,u_2\})(\epsilon, \{a,b,u_2\})
(\epsilon,\Sigma).\]
 Note that, since there is no event $\sigma\in \Sigma_o\cap \Sigma_a$ in this example, 
 there is no element of form $(\sigma,\gamma)$ in the information-flow.
\end{myexm}

Based on its information-flow, the intruder can also estimate the system state. Note that, since the intruder can only access the online control decision issued at each step, it has no \emph{a priori} knowledge of the supervisor’s observable event set \(\Sigma_o\) or controllable event set \(\Sigma_c\), let alone the supervisor’s control policy.
Therefore, from the intruder’s point of view, the system may be controlled by an arbitrary supervisor, as long as it can generate an information-flow compatible with the intruder’s observations. We denote by \(\mathbb{S}\) the set of all possible supervisors with arbitrary event sets \(\Sigma_o'\) and \(\Sigma_c'\). Then, the state estimate of the intruder, which only knows online-released control decisions, is defined as follows.

\begin{mydef}[\bf Controlled State Estimate]\upshape\label{def:on-curr-est}
Given system \(G\),  for any information-flow \(\alpha\in \mathbb{P}_S(\L(S/G))\) generated under some supervisor $S\in \mathbb{S}$, which is unknown a priori, 
 we define the intruder's \emph{controlled state estimate} of the closed-loop system as
    \begin{equation}\label{eq:def-cse}
        \mathcal{E}_{a}(\alpha)\!=\!\{\delta(s)\in X: \exists S'\in \mathbb{S}\text{ s.t. }\mathbb{P}_{S'}(s)=\alpha\}. 
    \end{equation}
\end{mydef}

We use the following example to illustrate the controlled state estimate  of the closed-loop system and compare it with the open-loop state estimate.

\begin{myexm}[Controlled State Estimate]\label{example2}
Let us  still consider system \(G\) in Figure~\ref{fig:motiexample} and assume that the system is controlled by  supervisor \(S\) in Example~\ref{example1}, which is unknown \emph{a priori} to the intruder. 
Suppose the string executed is \(a u_1 u_2 u_2\in \L(G)\).
Without the information of control decisions, 
the intruder observes \(P_a(a u_1 u_2 u_2)=a\). Then we have 
the open-loop state estimate
$\mathcal{E}_{a}^{op}(a)=\{1,2,3,4,5,6,7\}$.
By additionally observing the online released control decisions, 
the information-flow of the intruder is  
\(\mathbb{P}_S(a u_1 u_2 u_2)=(\epsilon,\Sigma)(a,\epsilon) (\epsilon,\{a, b, u_2\})(\epsilon,\{a, b, u_2\})(\epsilon,\Sigma)\)
and the controlled state estimate is $\mathcal{E}_{a}(\mathbb{P}_S(a u_1 u_2 u_2))=\{7\}$, i.e.,  the intruder knows for sure that the current state is \(7\). 
This is because three control decisions are released after the intruder observes event \(a\), implying that at least three events in $\Sigma_o$ have occurred after \(a\). According to the system structure, the current state must belong to \(\{6,7\}\). Furthermore, the second-to-last released control decision is \(\{a,b,u_2\}\), which indicates that events \(u_1\) and \(u_3\) were disabled before the intruder observes the final control decision \(\Sigma\). The latter must be issued upon the transitions from states \(4,5\) to \(6,7\). Therefore, the current state cannot be \(6\), since state \(6\) can only be reached via event \(u_1\).
\end{myexm}

Similarly, for the case where the intruder can access the online control decisions,  we say the closed-loop system is opaque if the intruder cannot, based on its information-flow, determine with certainty that the system is in a secret state.

\begin{mydef}[\bf Opacity with A Priori Unknown Supervisors]\upshape\label{def:online-cur-opa}    
Given system $G$ and supervisor  \(S\), 
we say that the closed-loop system $S/G$  is opaque (w.r.t. \(\Sigma_a\) and \(X_S\)) if
    \begin{align}
        (\forall s\in \L(S/G))[\mathcal{E}_{a}(\mathbb{P}_S(s))\not\subseteq X_S].
    \end{align}
\end{mydef}

Our objective is to synthesize a supervisor that restricts the system behavior so that the resulting closed-loop system is opaque under the \emph{a priori unknown supervisor} setting. Hereafter, we no longer consider the open-loop opacity case, and the term opacity always refers to the closed-loop setting.

\begin{myprob}[\bf Opacity-Enforcing Supervisory Control Problem]\upshape\label{prob:online-opa-enf}
Given open-loop system \(G\) and secret states \(X_S\subseteq X\), synthesize a supervisor \(S:P_o(\L(G))\to\Gamma\) such that the closed-loop system \(S/G\) is opaque with the synthesized a priori unknown \(S\) (w.r.t. \(\Sigma_a\) and \(X_S\)).
\end{myprob}

Before formally presenting our synthesis procedure, we provide a candidate solution to the running example to better illustrate the problem setting.

\begin{myexm}[Opacity-Enforcing Supervisor]\label{exm:opa-enforcing}
We still consider the running example with secret state \(X_S=\{7\}\). 
Clearly,  the supervisor \(S\) described in Example~\ref{example1} is not opaque since
\(\mathcal{E}_{a}(\mathbb{P}_S(a u_1 u_2 u_2))=\{7\}\subseteq X_S\).
A possible solution to enforce opacity is to modify \(S\) to \(S'\) by changing \(S(u_1u_2)=\{a,b,u_2\}\) to \(S'(u_1u_2)=\Sigma\), while keeping the remaining control decisions unchanged. To see how \(S'\) enforces opacity, consider the (unique) string \(au_1u_2u_2 \in L(S'/G)\)   leading to the secret state \(7\).
Along this string, the information-flow is 
\[
    \mathbb{P}_{S'}(au_1u_2u_2)=(\epsilon,\Sigma) (a,\epsilon)(\epsilon,\{a,b,u_2\}) (\epsilon,\Sigma)(\epsilon,\Sigma). 
\]
On the other hand,  for string $au_1u_2u_1\in\L(S'/G)$ leading to no-secret state $6\notin X_S$, we also  have 
    \[
    \mathbb{P}_{S'}(au_1u_2u_1)=(\epsilon,\Sigma) (a,\epsilon)(\epsilon,\{a,b,u_2\}) (\epsilon,\Sigma)(\epsilon,\Sigma).
    \] 
Therefore, we have 
\(\mathcal{E}_{a}(\mathbb{P}_{S'}(au_1u_2u_2))=\{6,7\}\not\subseteq X_S\)  
and \(S'/G\) is opaque under online-released control decisions.
\end{myexm}

\begin{remark}[Comparison with the A Priori Known Setting]\upshape
When the control policy of the supervisor \(S\) is fully known \emph{a priori} to the intruder, but the intruder cannot access the online control decisions issued by the supervisor, for any executed string \(s\in \mathcal{L}(S/G)\),  the state estimate of the closed-loop system from the intruder’s perspective is defined as
\[
\mathcal{E}_a^S(P_a(s))
=\{
\delta(t)\in X:
t\in \mathcal{L}(S/G)\wedge
P_a(s)=P_a(t)\}.
\]
In general, the two state estimates \(\mathcal{E}_a^S(P_a(s))\) and \(\mathcal{E}_a(\mathbb{P}_{S}(s))\) are incomparable, since they leverage different information: 
\begin{itemize}
    \item 
    The \emph{a priori unknown} setting allows the intruder to precisely observe the control decision currently being applied, but it does not know which control decisions will be issued in the future, as they may change through events that are unobservable to the intruder.
    \item 
    By contrast, in the \emph{a priori known} supervisor setting, the intruder can further leverage the supervisor’s complete control policy, but it may still face uncertainty about which specific control decision is currently applied.
\end{itemize}
Yet, when the intruder’s observable event set is larger than the supervisor’s, i.e., 
\(\Sigma_o \subseteq \Sigma_a\), these two state estimates coincide. 
We do not provide a formal proof here due to space constraints; however, the intuition is straightforward.
Specifically, under the assumption \(\Sigma_o \subseteq \Sigma_a\), the intruder knows precisely when the supervisor updates its control decision. Hence, if the supervisor’s control policy is known \emph{a priori}, the intruder can combine this supervisory model with its online observations to uniquely determine the current control decision being applied. On the other hand, since \(\Sigma_o \subseteq \Sigma_a\), the control decision cannot change within the intruder’s unobservable reach. Therefore, knowing the supervisor’s future decision logic offline brings no additional information compared with simply knowing the current decision online.
In other words, under \(\Sigma_o \subseteq \Sigma_a\), our problem coincides with the standard opacity supervisory control problem with an \emph{a priori} known supervisor. However, to our knowledge, existing results for this setting still require the additional assumption \(\Sigma_c \subseteq \Sigma_o\), which is no longer needed in our approach.\qed
\end{remark}

\section{Information Structures with A Priori Unknown Supervisors }\label{section-IS}
In this section, we investigate how information evolves in the closed-loop system under our setting. We first study the information evolution from the intruder’s point of view. We then propose a new information structure for control synthesis, over whose state space the supervisors will be realized.

\subsection{Information Evolution from The Intruder's Perspective}

We first investigate the evolution of information from  the intruder’s perspective. 
Suppose that the plant is controlled by  supervisor \(S:P_o(\L(G))\to\Gamma\). 
Under the setting where control decisions are released online at each decision instant, the intruder’s state estimate updates as follows:
\begin{itemize}
\item 
Initially, starting from \(x_0\), the intruder observes the initial control decision \(\gamma_0\in\Gamma\) and updates its state estimate to 
\begin{equation}
    \hat{q}_{a,0}=\text{UR}_{\gamma_0,a}(  \{x_0\} )
\end{equation} 
by considering all unobservable strings in $(\gamma_0\cap\Sigma_{ua})^*$.
\item 
Subsequently, depending on the observation  type it receives, the intruder proceeds in one of the following ways:
\begin{itemize}
\item [1)] 
If a new event  without control decision   is observed, i.e., \((\sigma_1,\epsilon)\in\mathcal{O}\), 
 it updates the state estimate to  
\begin{equation}\label{eq:case1}
\hat{q}_{a,1}=\text{UR}_{\gamma_0,a}(\text{NX}_{\sigma_1}(\hat{q}_{a,0}))
\end{equation} 
by only considering the occurrence of \(\sigma_1\) and keep the control decision \(\gamma_0\) unchanged;
\item [2)] 
If  a new event together with a control decision    is observed, i.e., \((\sigma_1,\gamma_1)\in\mathcal{O}\), then it knows that a new control decision \(\gamma_1\) is issued by the supervisor upon the occurrence of \(\sigma_1\).
Therefore, the intruder should update its state estimate to
\begin{equation}
\hat{q}_{a,1}=\text{UR}_{\gamma_1,a}(\text{NX}_{\sigma_1}(\hat{q}_{a,0})).
\end{equation} 
Compared with Eq.~\eqref{eq:case1}, the only difference is that the unobservable reach should be computed based on \((\gamma_1\cap\Sigma_{ua})^*\) rather than 
\((\gamma_0\cap\Sigma_{ua})^*\);
\item [3)] 
If only a new control decision    is observed, i.e., \((\epsilon,\gamma_1)\in\mathcal{O}\),
the intruder should update its state estimate  to 
\[
\hat{q}_{a,1}=\text{UR}_{\gamma_1,a}(\text{UR}^{+}_{\gamma_0,a}(\hat{q}_{a,0})).
\] 
Here, one must first account for the occurrence of at least one unobservable event of the intruder in
\((\gamma_0\cap\Sigma_{ua})^*\), and then further compose it with the unobservable reach within \((\gamma_1\cap\Sigma_{ua})^*\).
\end{itemize} 
\end{itemize} 
The above process continues indefinitely as new information  $(\sigma,\gamma)\in\mathcal{O}$ arrives, thereby inducing the  
recursive state estimate  from the perspective of the intruder.

To formalize the recursive computation, let $G$ be the system  and \(S: P_o(\L)\to\Gamma\) be a supervisor. 
For any string 
\(s=\sigma_1\cdots\sigma_n\in \mathcal{L}(S/G)\), 
we denote by $\hat{\mathcal{O}}=\Sigma^\epsilon   \times \Gamma $ the set of event-decision tuple, 
and define the 
\textbf{augmented string} of string $s$ (under supervisor $S$) by
\[
\hat{s}\!=\!
(\epsilon,S(\epsilon))
(\sigma_1,S(P_o(\sigma_1) ))
\cdots
(\sigma_n,S(P_o(\sigma_1\cdots\sigma_n) ))
\!\in\! \hat{O}^*.
\]
That is, each event is augmented with the current control decision. Note that not all event or decision components in the augmented string are observable to the intruder, and therefore may not trigger an update to the state estimate.
We propose the following finite structure to compute the recursive state estimate along augmented strings.
\begin{mydef}[\bf Intruder State Estimator]\upshape\label{def:estimator}\
Given system $G$, the state estimator of the intruder with a priori unknown supervisor is defined as a 4-tuple
\[
\mathfrak{A} = 
(M,\hat{\mathcal{O}}, f, m_0),
    \]
where:
\begin{itemize}
    \item 
    $M \subseteq  (X \times 2^X \times \Gamma)\cup\{m_0\}$  is the set of states, 
    where $m_0$ is the unique initial state; 
    \item 
    $\hat{\mathcal{O}}=\Sigma^\epsilon   \times \Gamma $ is the set of events; 
    \item \(f:M\times \hat{\mathcal{O}}\to M\) is the deterministic transition function defined as follows:
    \begin{itemize}[leftmargin=*,itemsep=0pt,parsep=0pt]
        \item[1)]  
        For  initial state $m_0$ and each control decision \(\gamma \in \Gamma\),  we define
    \begin{equation}\label{eq:9}
          f(m_0, (\epsilon,\gamma) )
       =  (x_0,\text{UR}_{\gamma,a}(\{x_0\}), \gamma).
    \end{equation}
        \item[2)]  
   For each state $m=(x,q,\gamma)\in M $, event  \(\sigma\in\Lambda(x)\cap\gamma\)
    and control decision  \(\gamma' \in \Gamma\),  we define 
   $f(m, (\sigma,\gamma') )=( \delta(x,\sigma), q', \gamma'       )$, 
   where 
   \begin{equation}\label{eq:10}
 \begin{aligned}
	q' \!=\! 
		\left\{
		\begin{array}{ll}
		    \text{UR}_{\gamma,a}(\text{NX}_\sigma(q))
            & \text{if }     \sigma \in \Sigma_a\setminus \Sigma_o  \\
			  \text{UR}_{\gamma',a}(\text{UR}^{+}_{\gamma,a}(q))
            & \text{if }     \sigma \in \Sigma_o\setminus \Sigma_a   \\
          \text{UR}_{\gamma',a}(\text{NX}_\sigma(q))
          & \text{if }     \sigma \in \Sigma_o\cap \Sigma_a \\
          q  & \text{if }     \sigma \in \Sigma_{uo}\cap  \Sigma_{ua} 
		\end{array}
		\right..\!\!   
\end{aligned} 
\end{equation}
\end{itemize}
\end{itemize}
\end{mydef}

Intuitively, for each state \((x, q, \gamma ) \in M\) in $\mathfrak{A} $, the first component \(x \in X\) represents the actual state of the plant, the second component \(q \in 2^X\) represents the intruder's state estimate, and the last component \(\gamma \in \Gamma\) represents the current control decision being applied.
Therefore, for each new event \((\sigma, \gamma') \in \hat{\mathcal{O}}\), \(\sigma \in \Sigma\) represents the actual event that occurs in the plant, and \(\gamma'\) represents the new control decision applied upon the occurrence of \(\sigma\). As a result, the plant's state component updates to \(\delta(x, \sigma) \) according to the system dynamics, while the last component records the new control decision \(\gamma'\).
The second component of the intruder's state estimate updates from \(q\) to \(q'\) based on the cases discussed above, categorized by event observability. Note that \(m_0\) is the initial state, from which only events of the form \((\epsilon, \gamma)\) are defined; this captures the fact that the supervisor makes an initial control decision before any events occur.

The following result establishes the correctness of the proposed state estimator, 
which states that, for any string, 
the state estimator of the intruder based on the information-flow, 
is exactly the second component of the state reached by the augmented string in the state estimator. 

\begin{mypro}\upshape\label{prop:1}
Let $G$ be the system 
and \(S: P_o(\L)\to\Gamma\) be a supervisor. 
For any string 
\(s\in \mathcal{L}(S/G)\),  we have 
\begin{equation}
 f(\hat{s})=(\delta(s), \mathcal{E}_{a}(\mathbb{P}_S(s)),S(P_o(s))).    
\end{equation}
\end{mypro}

\begin{proof}
See the Appendix. 
\end{proof}

We illustrate the intruder state estimator $\mathfrak{A}$ by the following example.

\begin{myexm}[Intruder's Information-Flow]
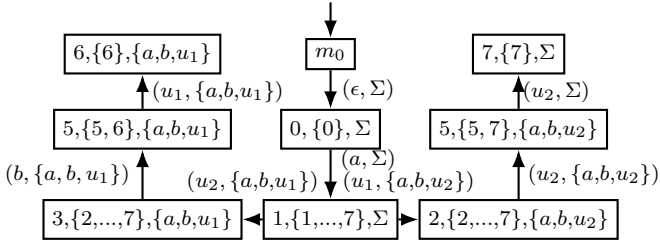
\begin{figure}
    \centering
        \begin{tikzpicture}[->,>={Latex}, thick, initial text={}, node distance= 1 cm, initial where=above, thick, O node/.style={rectangle, draw, minimum size=4mm, font=\footnotesize}]    
   \node[initial, state, O node] (0) { $ m_0$} ;  
   \node[state, O node, ] [below of = 0] (1) {\(0,\{0\},\Sigma\)};   
   \node[state, O node, ] [yshift=-2.2cm] (2) {\(1,\!\{1,\!...,\!7\},\!\Sigma\)};
   \node[state, O node, ] [xshift=-2.5cm,yshift=-2.2cm] (3) {\(3,\!\{2,\!...,\!7\},\!\{a,\!b,\!u_1\}\)};
   \node[state, O node, ] [xshift=2.5cm,yshift=-2.2cm] (4) {\(2,\!\{2,\!...,\!7\},\!\{a,\!b,\!u_2\}\)};
   \node[state, O node, ] [xshift=-2.5cm,yshift=-1cm] (5) {\(5,\!\{5,6\},\!\{a,\!b,\!u_1\}\)};
   \node[state, O node, ] [xshift=2.5cm,yshift=-1cm] (6) {\(5,\!\{5,7\},\!\{a,\!b,\!u_2\}\)};
   \node[state, O node, ] [above of=5] (7) {\(6,\!\{6\},\!\{a,\!b,\!u_1\}\)};
   \node[state, O node, ] [above of=6] (8) {\(7,\!\{7\},\!\Sigma\)};
   
   \path[->]
   (0) edge node [xshift=0.5cm] {\footnotesize{\((\epsilon,\Sigma)\)}} (1)
   (1) edge node [xshift=0.5cm,yshift=0.2cm] {\footnotesize{\((a,\Sigma)\)}} (2)
   (2) edge node [yshift=0.5cm] {\footnotesize{\((u_2,\{a,\!b,\!u_1\})\)}} (3)
   (2) edge node [yshift=0.5cm] {\footnotesize{\((u_1,\{a,\!b,\!u_2\})\)}} (4)
   (3) edge node [xshift=-1cm] {\footnotesize{\((b,\{a,b,u_1\})\)}} (5)
   (4) edge node [xshift=1cm] {\footnotesize{\((u_2,\{a,\!b,\!u_2\})\)}} (6)
   (5) edge node [xshift=1cm] {\footnotesize{\((u_1,\{a,\!b,\!u_1\})\)}} (7)
   (6) edge node [xshift=0.5cm] {\footnotesize{\((u_2,\Sigma)\)}} (8)
   ;
   \end{tikzpicture}
    \caption{Partial representation of online current-state estimator \(\mathfrak{A}\).}
    \label{fig:state-estimator}
\end{figure}
Let us still consider the system \(G\) as shown in Figure~\ref{fig:motiexample} and the supervisor \(S\) in Example~\ref{example1}. 
Part of the overall structure \(\mathfrak{A}\) is presented in Figure~\ref{fig:state-estimator}, where we only show states reached by augmented strings under the given supervisor $S$. 
Let us consider string  \(s=au_1u_2u_2\in \L(S/G)\). 
Then its augmented string is 
\[
\hat{s}=(\epsilon,\Sigma )(a,\Sigma)(u_1,\{a, b, u_2\})(u_2,\{a, b, u_2\})(u_2,\Sigma)
\]
and the information-flow of the intruder is 
\[
\mathbb{P}_S(a u_1 u_2 u_2)=(\epsilon,\Sigma) (a,\epsilon) (\epsilon,\{a, b, u_2\})(\epsilon,\{a, b, u_2\})(\epsilon,\Sigma).
\]
Upon the first event 
$(\epsilon,\Sigma) $,  
we have
$\text{UR}_{\Sigma,a}(\{0\})=\{0\}$, i.e., 
the state estimator reaches the state
$m_1=(0,\{0\},\Sigma)$. 
Then, upon event $(a,\Sigma)$ in $\mathfrak{A}$ 
(note that the intruder observes $(a,\epsilon)$ since no new control decision is issued), 
we have
\[
    \text{UR}_{\Sigma,a}(\text{NX}_a(\{0\}))=\text{UR}_{\Sigma,a}(\{1\})=\{1,2,3,4,5,6,7\}, 
    \] 
i.e., \(\mathfrak{A}\)  reaches state
$m_2=(1,\{1,2,3,4,5,6,7\},\Sigma)$. 
Next, upon event $ (\epsilon,\{a, b, u_2\})$, 
we have 
\begin{align}
    &\text{UR}_{\{a,b,u_2\},a}(\text{UR}^{+}_{\Sigma,a}(\{1,2,3,4,5,6,7\}))\nonumber\\
    =&\text{UR}_{\{a,b,u_2\},a}(\{2,3,4,5,6,7\})=\{2,3,4,5,6,7\}\nonumber
\end{align} 
and the state estimator reaches the state
$m_3\!=\!(2,\{2,3,4,5,6,7\},$ $\{a, b, u_2\})$. 
Similarly, 
upon events 
$(\epsilon,\{a, b, u_2\})$ and $(\epsilon,\Sigma)$, 
the state estimator reaches 
$m_4=(5,\{5,7\},\{a, b, u_2\})$ 
and 
$m_5=(7,\{7\},\Sigma)$ in sequence. 
As we can see, the second component of $m_5$ is exactly the state estimate of the intruder upon the information-flow
$\mathbb{P}_S(a u_1 u_2 u_2)$,  which aligns with our previous result in Example~\ref{example2}.
\end{myexm}

\subsection{Information Structure of the Supervisor}
Now, we tackle the main problem of what information structure is needed to realize an opacity-enforcing supervisor. Since the observations of the supervisor and the intruder are generally incomparable, the supervisor also remains uncertain about the intruder's knowledge. Therefore, the supervisor must estimate the intruder’s state estimate from its own perspective.
To this end, we define a set of the intruder’s state estimates as an \textbf{information state}. Thus, 
\begin{equation}
\mathbb{I} = 2^M = 2^{(X \times 2^X \times \Gamma)\cup\{m_0\}}    
\end{equation}
is the set of information states.
Specifically, 
an information state $\imath\in \mathbb{I}$ is of form
\begin{equation}\label{eq:ISdef}
    \imath=\{ \underbrace{(x_1,q_1,\gamma_1)}_{m_1},\underbrace{(x_2,q_2,\gamma_2)}_{m_2},\dots, \underbrace{(x_k,q_k,\gamma_k)}_{m_k}  \}.
\end{equation}
We denote by 
\begin{align}
    X(\imath)=\{x_1,\dots,x_k\}
\text{ and }
    Q(\imath)=\{q_1,\dots,q_k\}
\end{align}
the set of plant states 
and the set of state estimates in $\imath$.
We say an information state $\imath\in \mathbb{I}$ of form Eq.~\eqref{eq:ISdef}
is \emph{consistent} if $\gamma_i=\gamma_j,\forall i,j$. 
If $\imath$ is consistent, then we denote by $\Gamma(\imath)$ the unique control decision in each state estimate within $\imath$. 
Since we will use the information state to represent the knowledge of the supervisor, and both the supervisor and the intruder always know the control decision applied, all information states encountered hereafter are consistent by nature.

Now, let us evaluate the evolution of the supervisor's knowledge, i.e., its information state, upon each new piece of information received.
Suppose that the supervisor's current information state is $\imath\in \mathbb{I}$ of form Eq.~\eqref{eq:ISdef}. 
When it observes a new event $\sigma\in \Sigma_o$ and issues a control decision $\gamma\in \Gamma$, 
the supervisor should update its knowledge of the intruder’s estimate by considering how each possible intruder’s state estimate is updated based on the same information. This is captured by the following operator:
\begin{equation}
   \mathbb{NX}_{(\sigma,\gamma)}(\imath)
=\{  m'\!\in\! M: \exists m\!\in\!    \imath\text{ s.t. } m'\!=\!f(m,   (\sigma,\gamma) )  \}.
\end{equation}
The above knowledge update is only one step, as the intruder may further observe events in $\Sigma_a\setminus \Sigma_o$. 
Therefore, the supervisor should also account for the subsequent estimate updates of the intruder that are silent from its own perspective.
To this end, for any information state   $\imath\in \mathbb{I}$, 
we further define the following operator: 
\begin{equation}
   \mathbb{UR}_\gamma(\imath)
=\left\{  m'\!\in\! M: \!\!\!\!\!\!\!\!
\begin{array}{cc}
     &  \exists m\!\in\! \imath, w\!\in\!  (  ( \Sigma_{uo}\cap\gamma)  \times 
     \{\gamma\}    )^*\\
     & \text{ s.t. }  
m'\!=\!f(m,  w )
\end{array}\!\!\!\!
\right\}. 
\end{equation}
Intuitively, for each possible intruder's state estimate \(m\), before the supervisor observes the next event, strings in \((\Sigma_{uo} \cap \gamma)^*\) may occur. These events are unobservable to the supervisor but have been enabled. Note that the control decision within the unobservable reach remains unchanged. Therefore, all unobservable events are augmented with the same decision \(\gamma\) in \(w\). This ensures that the supervisor maintains consistency in the control decisions applied, even though it cannot observe these unobservable events.

In what follows, we restrict our attention to the IS-based supervisor, which makes decisions solely based on the current information state rather than the entire history. 
Under such supervisors, the evolution of the information state can be uniquely determined using finite memory. To formalize this, we introduce the following IS-based control structure.

\begin{mydef}[\bf IS-based Control Structures]\upshape\label{def:IS-Cstructure}
An information-state-based (IS-based) control structure is a tuple
\[
    \S=(\mathbb{I}^\S_D, \mathbb{I}^\S_O,h^\S_{OD},h^\S_{DO},\imath^\S_0),
\]
where
\begin{itemize}
 \item 
 \(\mathbb{I}^\S_O \subseteq\mathbb{I}\) is the set of observation-states;
 \item 
 \(\mathbb{I}^\S_D \subseteq\mathbb{I}\times \Sigma^\epsilon_o\)  is the set of decision-states;
 \item 
 $h^\S_{OD}: \mathbb{I}^\S_O\times \Sigma_o\to \mathbb{I}^\S_D$
 is the transition function from observation-states to decision-states 
 such that:
 for every \(\imath\in\mathbb{I}^\S_O\) and \(\sigma\in\Sigma_o\cap \Gamma(\imath)\), 
 we define
 \begin{equation}
  h^\S_{OD}(\imath,\sigma)=   (\imath, \sigma).
 \end{equation}
 \item 
 $h^\S_{DO}: \mathbb{I}^\S_D\times \Gamma \to \mathbb{I}^\S_o$
 is the transition function from decision-states to observation-states 
 such that: 
 for every  \((\imath,\sigma)\in \mathbb{I}^\S_D \), 
 there exists a unique decision \(\gamma\in \Gamma\) such that \(h^\S_DO(\imath,\gamma)!\)
 and 
 \begin{equation}
    h^\S_{DO}(\imath,\gamma)=\mathbb{UR}_{\gamma}(\mathbb{NX}_{(\sigma,\gamma)}(\imath))
 \end{equation} 
 \item 
 $\imath^\S_0=(  \{ m_0 \}    ,\epsilon)\in \mathbb{I}^\S_D$ is the initial state, which is a decision-state.
    \end{itemize}
\end{mydef}

The IS-based control structure operates intuitively as follows. 
From each decision-state \(\mathbb{I}^\S_D\), 
there exists a unique control decision defined, upon which the structure moves to a successor observation-state \(\imath'\in \mathbb{I}^\S_O\) via transition \(h^\S_{DO}\). 
Then, from this observation-state,
all feasible observations  \(\sigma\in\Sigma_o\cap\Gamma(\imath')\) are defined, 
and each occurrence takes the structure to another   decision-state according to the transition \(h^\S_{OD}\), and so forth.
Following this, 
for each observation string 
\(\alpha=\sigma_1\sigma_2...\sigma_n\in\Sigma_o^*\), 
it induces a unique sequence in the control structure
\vspace{-3pt}
\begin{equation}\label{eq:flow}
    \imath^\S_0\xrightarrow{\gamma_0}\imath'_0\xrightarrow{\sigma_1}\imath_1\xrightarrow{\gamma_1}\cdots
   \xrightarrow{\gamma_{n-1}}\imath'_{n-1}\xrightarrow{\sigma_n}\imath_n \xrightarrow{\gamma_{n}}\imath'_{n},
\end{equation}
where $\gamma_i$ is the unique control decision defined at decision-state $\imath_i$.  
We denote 
$\mathbb{I}^\S_D(\alpha)=\imath_n$ and  
$\mathbb{I}^\S_O(\alpha)=\imath'_n$ as
the decision-state and the observation-state reached upon $\alpha$ in $\S$, respectively.
This essentially allows us to decode an IS-based supervisor from each IS-based control structure. 

\begin{mydef}[\bf Decoded Supervisors]\upshape
    Given an IS-based control structure  \(\S\),  its decoded supervisor $S:P_o(\mathcal{L}(G))\to \Gamma$ is defined by: \vspace{-3pt}
\begin{equation}
\forall \alpha=\sigma_1\sigma_2\dots \sigma_n\in P_o(\mathcal{L}(G)): 
S(\alpha)= \gamma_{n}, 
\end{equation}
where $\gamma_n$ is the unique control decision at $\mathbb{I}^\S_D(\alpha)$ in $\S$.
\end{mydef}

The following result establishes the meaningfulness of the proposed information structure. Specifically, it shows that, for any supervisor decoded from an IS-based structure, the observation-state reached explicitly captures both the supervisor's own knowledge of the plant states and 
its knowledge about the state estimate from the intruder's perspective. This ensures that the supervisor’s decision-making process properly accounts for the uncertainty in the intruder’s knowledge and that the supervisor’s state representation aligns with the information-flow from both the supervisor’s and intruder’s viewpoints.

\begin{mypro}\upshape\label{prop:IS}
    Given an IS-based control structure \(\S\), let \(S:P_o(\mathcal{L}(G))\to\Gamma\) be its decoded supervisor. 
    For any observation $\alpha\in P_o(\L(S/G))$,  
    let $\imath=\mathbb{I}^\S_O(\alpha)$ be the observation-state reached. Then 
    we have \vspace{-3pt}
\begin{align}
    X(\imath)=&\ \mathcal{E}_{o}^S(\alpha)   \\
    Q(\imath)=&\ \{\mathcal{E}_{a}( \mathbb{P}_S(s) ):   
    s\in P_o^{-1}(\alpha)\!\cap \!\L(S/G)\}. \nonumber
\end{align} 
\end{mypro}
\begin{proof}
See the Appendix. 
\end{proof}
We illustrate the notions of the IS-based control structure and the decoded supervisor with the following example.

\begin{myexm}[\bf Control Structure and Decoded Supervisor]

\begin{figure}
    \centering
        \begin{tikzpicture}[->,>={Latex}, thick, initial text={}, node distance= 1 cm, initial where=above, thick, D node/.style={rectangle, rounded corners, draw, minimum size=4mm, font=\footnotesize}, O node/.style={rectangle, draw, minimum size=4mm, font=\footnotesize}]    
   \node[initial, state, D node] (0) { $ \{m_0\},\epsilon $ } ;  
   \node[state, O node, ] [below of = 0] (1) {\(\{\!(0,\{0\}\!),\!(1,\!\{1,...,7\}\!)\!\}, \Sigma\)};   
   \node[state, D node][below of = 1] (2) {\(\{\!(0,\{0\}\!),\!(1,\!\{1,...,7\}\!)\!\},\! \Sigma,\!u_2\)};
   \node[state, D node][below of = 2] (3) {\(\{\!(0,\{0\}\!),\!(1,\!\{1,...,7\}\!)\!\}, \!\Sigma,\!u_1\)};
   \node[state, O node][below of = 3] (4) {\(\{\!(2,\{2,...,7\}\!)\!\}, \{a,\!b,\!u_2\}\)};
   \node[state, D node][below of = 4] (5) {\(\{\!(2,\{2,...,7\}\!)\!\}, \{a,\!b,\!u_2\},u_2\)};
   
   \node[state, O node][xshift=4.4cm] (6) {\(\{\!(6,\{6\}\!)\!\}, \{a,\!b,\!u_1\}\)};
   \node[state, D node][below of = 6] (7) {\(\{\!(3,\{2,\!...,\!7\}\!),\!(5,\{5,\!6\}\!)\!\},\! \{a,\!b,\!u_1\},\!u_1\)};
   \node[state, O node][below of = 7] (8) {\(\{\!(3,\{2,\!...,\!7\}\!),\!(5,\{5,\!6\}\!)\!\},\! \{a,\!b,\!u_1\}\)};
   \node[state, O node][below of = 8] (9) {\(\{\!(7,\{7\}\!)\!\}, \Sigma\)};
   \node[state, D node][below of = 9] (10) {\(\{\!(5,\{5,7\}\!)\!\}, \{a,\!b,\!u_2\},u_2\)};
   \node[state, O node][below of = 10] (11) {\(\{\!(5,\{5,7\}\!)\!\}, \{a,\!b,\!u_2\}\)};
   
    \draw[->] (1) -- (-2.2cm, -1cm)--(-2.2cm, -3cm)--(3);
   
   \path[->]
   (0) edge node [] {} (1)
   (1) edge node [] {} (2)
   (3) edge node [] {} (4)
   (4) edge node [] {} (5)
   (7) edge node [] {} (6)
   (8) edge node [] {} (7)
   (2) edge node [] {} (8)
   (10) edge node [] {} (9)
   (11) edge node [] {} (10)
   (5) edge node [] {} (11);
   \end{tikzpicture}
    \caption{IS-based Control Structure of \(S\) in Example~\ref{example1}.}
    \label{fig:controlstructure}
\end{figure}
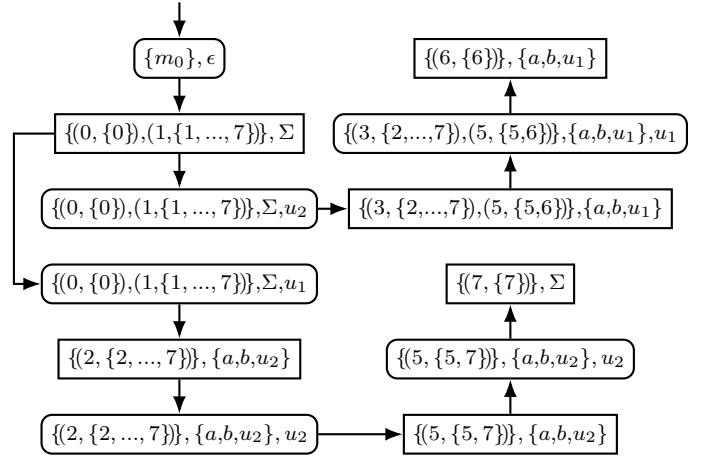

We still consider the system shown in Figure~\ref{fig:motiexample}. 
An example of the IS-based control structure is shown in Figure~\ref{fig:controlstructure}. 
In this structure, each rectangle with rounded corners represents a decision-state \(\imath \in \mathbb{I}^\S_D\), from which a unique control decision $\gamma$ is selected; 
each plain rectangle represents an observation-state \(\imath' \in \mathbb{I}^\S_O\) from which transitions lead by feasible events \(\sigma\in\Sigma_o\cap\Gamma(\imath')\) are all defined. For simplicity, we factor out the consistent control decision for each \(m\in\imath\) and \(\imath\in\mathbb{I}\) and write simply as
\(
\imath=(\{(x_1,q_1),...,(x_k,q_k)\},\gamma)
\)
in Figure~\ref{fig:controlstructure}.

We start from the decision-state $\imath_0 \!=\! (\{m_0\},\epsilon)$, where the control decision is $\gamma_0 = \Sigma$ (all events enabled), the system evolves to the observation-state $\imath'_0={\mathbb{UR}}_\Sigma({\mathbb{NX}}_{(\epsilon,\Sigma)}(\imath_0))=\mathbb{UR}_\Sigma(\{\!(0,\{0\}\!,\Sigma)\!\})=\{\!(0,\{0\},\!\Sigma),\!(1,\{1,...,7\},\!\Sigma)\!\}$. 
Then, upon observation \(u_2\), the structure transitions to decision-state $\imath_1 = (\imath'_0,u_2)$, at which the unique control decision $\gamma_1 = \{a,b,u_1\}$ (events \(u_2,u_3\) disabled) is applied, and leading deterministically to the next observation-state \(\imath'_1=(\{\!(3,\{2,...,7\},\!\{a,\!b,\!u_1\}),\!(5,\{5,6\},\!\{a,\!b,\!u_1\})\!\})\). After this, the structure transitions to \((\{\!(6,\{6\}\!)\!\}, \{a,\!b,\!u_1\})\) via observation \(u_1\) and applies the decision \(\{a,b,u_1\}\). Another branch is defined analogously.
This control structure induces a unique supervisor \(S\) defined by: \(S(u_1)=S(u_1u_2)=\{a,b,u_2\}\), \(S(u_2)=S(u_2u_1)=\{a,b,u_1\}\), and \(S(\alpha)=\Sigma\) for any other \(\alpha\in P_o(\L(G))\), which is exactly the one discussed in Example~\ref{example1}.
\end{myexm}

\section{Synthesis of IS-Based Supervisors}\label{sec:alg}
In this section, we tackle the supervisor synthesis problem based on the proposed information structure. We first present an algorithm to effectively search for an IS-based supervisor that enforces opacity. We then prove the soundness and completeness of the algorithm by showing that restricting attention to IS-based supervisors is without loss of generality.

\subsection{Supervisor Synthesis Algorithm}\label{section-alg}

Recall that, in the construction of the IS-based control structure, we showed that the second component of the observation-state reach is the set of all possible state estimates of the intruder that are consistent with the supervisor's observation. For any observation-state \(\imath \in \mathbb{I}^\S_O\), we say that $\imath$ is \textbf{safe} if it does not contain any element that is a possible state estimate of the intruder, entirely composed of secret states. Formally, this condition is given by:
        \begin{align}
            \forall q \in Q(\imath): q \not\subseteq X_S.
        \end{align}
We say an IS-based control structure is safe if all observation-states in it are safe. 
The following result shows that, to enforce opacity, it suffices to construct an IS-based control structure in which all observation-states are safe.

\begin{mythm}\upshape\label{theo:sound}
    Let \(\S\) be a safe IS-based control structure. Then its decoded supervisor \(S\) solves Problem~\ref{prob:online-opa-enf}, i.e., 
    \[
    (\forall s\in \L(S/G))[\mathcal{E}_{a}(\mathbb{P}_S(s))\not\subseteq X_S].  
    \]
\end{mythm}

\begin{proof}
We prove by contraposition. 
Suppose that \(S/G\) is not opaque.  Then we know that there exists a 
string \(s\in\L(S/G)\), such that \(\mathcal{E}_{a}(\mathbb{P}_S(s^*))\subseteq X_S\). 
According to Proposition~\ref{prop:IS}, 
for observation \(\alpha=P_o(s)\), we have \(q=\mathcal{E}_{a}(\mathbb{P}_S(s))\in Q(\mathbb{I}^\S_O(\alpha))\) such that \(q\subseteq X_S\). Therefore, observation-state \(\mathbb{I}^\S_O(\alpha)\) is not safe, which means that  \(\S\) is not safe.
\end{proof}

In order to find a safe IS-based control, we use a safety-game-based approach, 
which is formally provided in  Algorithm~\ref{alg:main}. 
The main idea  is to first enumerate all possible control decisions and then extract a feasible \(\S\) from them. Specifically, Algorithm~\ref{alg:main} consists of the following three steps.

\begin{algorithm}[t]
\caption{Synthesis of Control Structure \(\S\)}\label{alg:main} 
\KwIn{\(G\),\  \(\Sigma_c\),\  \(\Sigma_a\),\  \(\Sigma_o\),\  \(X_S\)}
\KwOut{\(\S\)}

\(\imath^\mathfrak{B}_0\gets (\{m_o\},\epsilon), \ \ \imath\gets\imath^\mathfrak{B}_0 \)
\\
\(\mathbb{I}^\mathfrak{B}_D\gets\{\imath^\mathfrak{B}_0\}\),
\ \
\(\mathbb{I}^\mathfrak{B}_O\gets\emptyset\),
\ \
$h_{DO}^\mathfrak{B}\gets \emptyset$, 
\ \ 
$h_{OD}^\mathfrak{B}\gets \emptyset$
\\
\ForEach{\(\gamma\in\Gamma\)}
{
    \If{\(\imath'=h^\mathfrak{B}_{DO}(\imath,\gamma)\) is safe}
    {
        \(h^\mathfrak{B}_{DO}\gets h^\mathfrak{B}_{DO}\cup \{(\imath,\gamma,\imath')\}\) \\
        \If{\(\imath'\not\in\mathbb{I}^\mathfrak{B}_O\)}
        {
           \(\mathbb{I}^\mathfrak{B}_O\gets\mathbb{I}^\mathfrak{B}_O\cup\{\imath'\}\) \\
           \ForEach{\(\sigma\in\Sigma_o\cap\gamma\)}
           {
                \(h^\mathfrak{B}_{OD}\gets h^\mathfrak{B}_{OD}\cup \{(\imath',\sigma,h^\mathfrak{B}_{OD}(\imath',\sigma))\}\) \\
                \If{\(h^\mathfrak{B}_{OD}(\imath',\sigma)\not\in\mathbb{I}^\mathfrak{B}_D\)}
                {
                    \(\imath\gets h^\mathfrak{B}_{OD}(\imath',\sigma)\), \ \
                    \(\mathbb{I}^\mathfrak{B}_D\gets \mathbb{I}^\mathfrak{B}_D\cup \{\imath\}\) \\
                    back to line~3
                }
           }
        }
    }
}
\While{\(\mathbb{I}_{bad}^\mathfrak{B}\neq\emptyset\)}
{
Remove all incomplete states \(\mathbb{I}_{bad}^\mathfrak{B}\) from \(\mathfrak{B}\)\\
Remove all transitions involving \(\mathbb{I}_{bad}^\mathfrak{B}\) from \(\mathfrak{B}\) 
}
\eIf{\(\imath^\mathfrak{B}_0\not\in\mathbb{I}^\mathfrak{B}_D\)}
{
\textbf{return} no solution exists}
{
\(\mathbb{I}^\S_D \gets \{ \imath^\mathfrak{B}_0 \}\), \(\mathbb{I}^\S_O \gets \emptyset\), \(h^\S_{DO} \gets\emptyset\), \(h^\S_{OD} \gets\emptyset\), \(\imath\gets\imath^\mathfrak{B}_0\) \\
Pick a transition \((\imath,\gamma,\imath')\in h^\mathfrak{B}_{DO}\) \\
\(h^\S_{DO}\gets h^\S_{DO}\cup \{(\imath,\gamma,\imath')\}\) \\
\If{\(\imath'\not\in \mathbb{I}^\S_O\)}
{
    \(\mathbb{I}^\S_O\gets\mathbb{I}^\S_O\cup\{\imath'\}\) \\
    \ForEach{\((\imath',\sigma,\imath'')\in h^\mathfrak{B}_{OD}\)}
    {
        \(h^\S_{OD}\gets h^\S_{OD} \cup \{(\imath',\sigma,\imath'')\}\) \\
        \If{\(\imath''\not\in\mathbb{I}^\S_D\)}
        {
            \(\imath\gets \imath'', \ \ \mathbb{I}^\S_D\gets \mathbb{I}^\S_D\cup\{\imath\}\) \\
            back to line~20
        }
    }
}
\textbf{return} IS-based Control Structure \(\S\)} 
\end{algorithm}

\textbf{Step 1: Initial Expansion (line 1--12).}
This step constructs a structure \vspace{-3pt}
\[
\mathfrak{B} = \bigl( \mathbb{I}^\mathfrak{B}_D,\ \mathbb{I}^\mathfrak{B}_O,\ h^\mathfrak{B}_{DO},\ h^\mathfrak{B}_{OD},\ \imath^\mathfrak{B}_0 \bigr)\vspace{-3pt}
\]
that enumerates all reachable information-states satisfying safety. 
Compared to the IS-based control structure in Definition~\ref{def:IS-Cstructure}, 
the decision-state in \(\mathfrak{B}\) may have \emph{multiple} control decisions.
Specifically, lines 1-2 initialize \(\mathfrak{B}\) with a single decision-state:
\(( \{m_0\},\epsilon)\).
A depth-first search is then performed to recursively expand \(\mathfrak{B}\) from the initial state, exploring only observation-states that remain safe throughout the process.
The expansion stops whenever it encounters either: (i) an observation-state violates safety, or
(ii) an information-state has been previously visited.
By this construction, every state in \(\mathbb{I}^\mathfrak{B}_O\) preserves the required safety properties.

\begin{myexm}[Expand the Solution Space]
    \begin{figure*}
    \centering
        \begin{tikzpicture}[->,>={Latex}, thick, initial text={}, node distance= 1 cm, initial where=above, thick, D node/.style={rectangle, rounded corners, draw, minimum size=4mm, font=\footnotesize}, O node/.style={rectangle, draw, minimum size=4mm, font=\footnotesize}]    
   \node[initial, state, D node, color=blue] (0) { $ \{m_0\},\epsilon $ } ;  
   \node[state, O node, color=blue] [xshift=4cm] (1) {\(\{\!(0,\{0\}\!),\!(1,\!\{1,...,7\}\!)\!\}, \Sigma\)};   
   \node[state, D node, color=blue][below of = 1] (2) {\(\{\!(0,\{0\}\!),\!(1,\!\{1,...,7\}\!)\!\},\! \Sigma,\!u_2\)};
   \node[state, D node, color=blue][below of = 0] (3) {\(\{\!(0,\{0\}\!),\!(1,\!\{1,...,7\}\!)\!\}, \!\Sigma,\!u_1\)};
   \node[state, O node, color=blue][below of = 3] (4) {\(\{\!(2,\{2,...,7\}\!)\!\}, \{a,\!b,\!u_2\}\)};
   \node[state, D node, color=blue][below of = 4] (5) {\(\{\!(2,\{2,...,7\}\!)\!\}, \{a,\!b,\!u_2\},u_2\)};
   \node[state, D node, color=blue][xshift=13.2cm, yshift=-3cm] (7) {\(\{\!(3,\{2,\!...,\!7\}\!),\!(5,\{5,\!6\}\!)\!\},\! \{a,\!b,\!u_1\},\!u_1\)};
   \node[state, O node, color=blue][below of = 7] (6) {\(\{\!(6,\{6\}\!)\!\}, \Sigma\)};
   \node[state, O node, color=blue][xshift=8.4cm, yshift=-2cm] (8) {\(\{\!(3,\{2,\!...,\!7\}\!),\!(5,\{5,\!6\}\!)\!\},\! \{a,\!b,\!u_1\}\)};
   \node[state, O node, fill=red][xshift=4cm, yshift=-5cm] (9) {\(\{\!(7,\{7\}\!)\!\}, \Sigma\)};
   \node[state, O node, color=red][below of = 5] (11) {\(\{\!(5,\{5,7\}\!)\!\}, \{a,\!b,\!u_2\}\)};
   \node[state, D node, color=red][below of = 11] (10) {\(\{\!(5,\{5,7\}\!)\!\}, \{a,\!b,\!u_2\},u_2\)};
   \node[state, O node][above of=8] (12) {\(\{\!(3,\{2,\!...,\!7\}\!),\!(5,\{5,\!6\}\!)\!\},\! \Sigma\)};
   \node[state, O node, color=red][above of = 12] (13) {\(\{\!(3,\{2,\!...,\!7\}\!),\!(5,\{5,\!6\}\!)\!\},\! \{a,\!b,\!u_2\}\)};
   \node[state, D node][above of=7] (14) {\(\{\!(3,\{2,\!...,\!7\}\!),\!(5,\{5,\!6,\!7\}\!)\!\},\! \Sigma,\!u_1\)};
   \node[state, D node][above of=14] (15) {\(\{\!(3,\{2,\!...,\!7\}\!),\!(5,\{5,\!6,\!7\}\!)\!\},\! \Sigma,\!u_2\)};
   \node[state, D node, color=red][above of=15] (16) {\(\{\!(3,\{2,\!...,\!7\}\!),\!(5,\{5,\!7\}\!)\!\},\! \{a,\!b,\!u_2\},\!u_2\)};
   \node[state, O node, ][below of = 2] (17) {\(\{\!(5,\{5,6,7\}\!)\!\}, \{a,\!b,\!u_1\}\)};
   \node[state, O node, color=blue][below of = 17] (18) {\(\{\!(5,\{5,6,7\}\!)\!\}, \Sigma\)};
   \node[state, D node, color=blue][below of = 18] (19) {\(\{\!(5,\{5,6,7\}\!)\!\}, \Sigma,u_2\)};
   \node[state, D node, ][below of = 8] (20) {\(\{\!(5,\{5,6,7\}\!)\!\}, \!\{a,\!b,\!u_1\},\! u_1\)};
   \node[state, D node, color=blue][below of = 20] (21) {\(\{\!(5,\{5,6,7\}\!)\!\}, \Sigma,u_1\)};
   \node[state, O node, color=blue][below of = 21] (22) {\(\{\!(7,\{6,7\}\!)\!\},\! \Sigma\)};
   \node[state, O node, color=blue][below of = 6] (23) {\(\{\!(6,\{6,7\}\!)\!\},\! \Sigma\)};

   \draw[->, dashed] (16) -- (16cm, 0cm)--(16cm, -5.9cm)--(4cm, -5.8cm)--(9);
   \draw[->] (15) -- (15.9cm, -1cm)--(15.9cm, -5.7cm)--(8.4cm, -5.7cm)--(22);
   \draw[->] (14) -- (15.8cm, -2cm)--(15.8cm, -5cm)--(23);
   \draw[->, dashed] (10)--(9);
   
   \path[->]
   (0) edge node [] {} (1)
   (1) edge node [] {} (2)
   (1) edge node [] {} (3)
   (3) edge node [] {} (4)
   (4) edge node [] {} (5)
   (7) edge node [] {} (6)
   (8) edge node [] {} (7)
   (2) edge node [] {} (8)
   (2) edge node [] {} (12)
   (2) edge node [] {} (13)
   (11) edge node [] {} (10)
   (12) edge node [] {} (14)
   (12) edge node [] {} (15)
   (13) edge node [] {} (16)
   (5) edge node [] {} (11)
   (5) edge node [] {} (17)
   (5) edge node [] {} (18)
   (18) edge node [] {} (19)
   (18) edge node [] {} (21)
   (21) edge node [] {} (23)
   (17) edge node [] {} (20)
   (19) edge node [] {} (22)
   (20) edge node [] {} (6);
   \end{tikzpicture}
    \caption{Partial representation of structure \(\mathfrak{B}\). Incomplete states are highlighted in red. An extracted control structure is highlighted in blue.}
    \label{fig:alg}
\end{figure*}
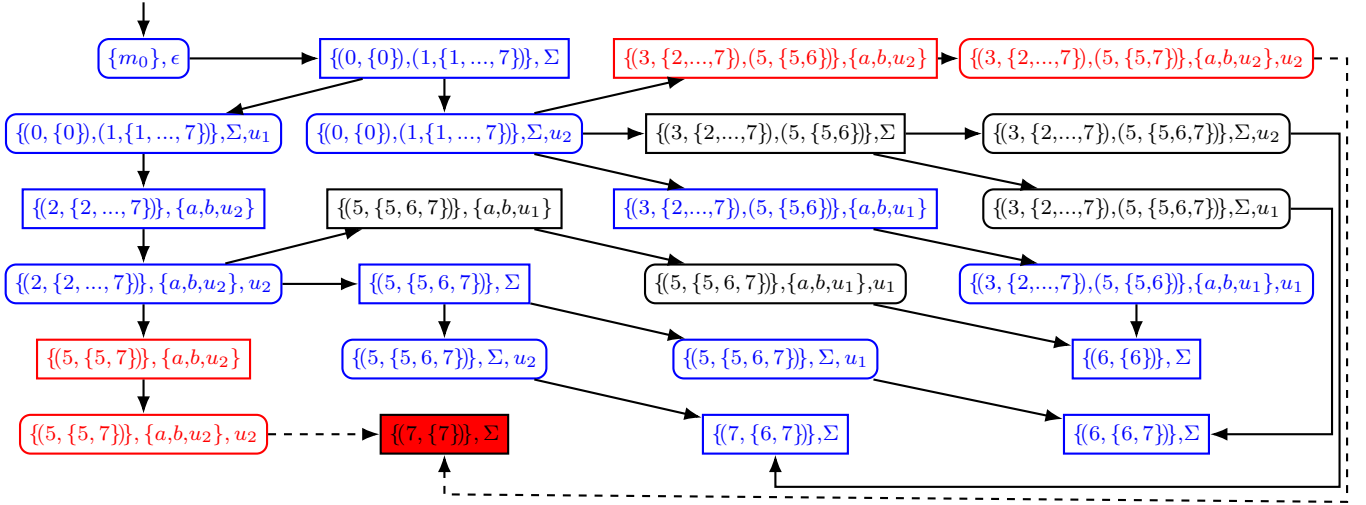
Let us continue with the running example. A portion of the overall structure \(\mathfrak{B}\) is illustrated in Figure~\ref{fig:alg}. 
Initially, we set the initial state as \(\imath^\mathfrak{B}_0 = (\{m_0\},\epsilon)\). 
Starting from here, the search iteratively expands \(\mathfrak{B}\) along all possible decisions; here we take the decision \(S(\epsilon)=\Sigma\) as an example. Then from the observation-state \(\{\!(0,\{0\},\!\Sigma),\!(1,\!\{1,...,7\},\!\Sigma)\!\}\), the search expands \(\mathfrak{B}\) along all possible observations \(u_1\) and \(u_2\) and so forth. This search terminates until it encounters either a state that has been previously visited, which are states \(\{(7,\{6,7\},\!\Sigma)\!\}\), \(\{(6,\{6,7\},\!\Sigma)\!\}\) and \(\{(6,\{6\},\!\Sigma)\!\}\) in the shown part, or the unsafe observation-state, which is \(\{(7,\{7\},\!\Sigma)\!\}\) in the shown part. Thus, transitions drawn with the dashed line do not belong in \(\mathfrak{B}\).
\end{myexm}

\textbf{Step 2: Completeness Check (line 13--15).}
This step prunes the previously obtained structure \(\mathfrak{B}\) by removing the following two types of  \emph{incomplete} states.
\begin{itemize}
    \item A decision-state \(\imath \in \mathbb{I}^\mathfrak{B}_D\) is incomplete if it has no control decision defined. 
    Thus, once the system reaches \(\imath\), no control decision can be taken to ensure safety. The set of incomplete decision-states is defined by  
    \[
        \mathbb{I}_{bad,D}^\mathfrak{B}=
        \{
        \imath \in \mathbb{I}^\mathfrak{B}_D
        :
        \forall \gamma\in \Gamma\text{ s.t. }f^\mathfrak{B}_{DO}(\imath,\gamma)\neg! 
        \}. 
    \]
    \item 
    An observation-state \(\imath \in \mathbb{I}^\mathfrak{B}_O\) is incomplete if some feasible observations are not defined. Since observations occur uncontrollably, reaching such a state would risk violating safety upon the occurrence of any missing observations. 
    The set of incomplete observation-states is defined  \vspace{-3pt}
    \begin{equation}   
    \mathbb{I}_{bad,O}^\mathfrak{B}=\left\{
        \imath \in \mathbb{I}^\mathfrak{B}_O:
        \!\!\!\!\!\!\! \begin{array}{cc}
             &   \exists \sigma \in \Sigma_o \cap \Gamma(\imath)\cap \Lambda(X(\imath))\\ 
             & \text{ s.t. }
              f^\mathfrak{B}_{OD}(\imath,\sigma)\neg!
        \end{array}
    \right\}.\nonumber
    \end{equation}
\end{itemize}

In this step, all incomplete states \(\mathbb{I}_{bad}^\mathfrak{B}=\mathbb{I}_{bad,D}^\mathfrak{B}\cup\mathbb{I}_{bad,O}^\mathfrak{B}\) together with their outgoing and incoming transitions are eliminated. 
However, the pruning of an incomplete state may cause a predecessor state to become incomplete. 
We thus perform this elimination in a while-loop until no further incomplete states remain.

\begin{myexm}[Prune Incomplete States]
  We continue with the running example. In Figure~\ref{fig:alg}, all incomplete states are highlighted in red. For example, state \((\{\!(5,\{5,7\},\!\{a,\!b,\!u_2\})\},u_2)\)  is incomplete because it has no feasible control decision defined. 
  By removing this red-highlighted state from \(\mathfrak{B}\), the state \(\{\!(5,\{5,7\},\!\{a,\!b,\!u_2\})\!\}\) becomes incomplete because of the absence of the observable events \(u_2\), and will therefore be removed in a subsequent iteration. After this iterative pruning process, the structure's completeness is finally maintained.
\end{myexm}

\textbf{Step 3: Supervisor Extraction (line 16--28).}
The algorithm finally constructs a control structure \(\S\) based on the remaining structure \(\mathfrak{B}\).
Specifically, we proceed as follows:
\begin{itemize}
    \item The initial decision-state of \(\S\) is set to be \(\imath=\imath_0^\mathfrak{B}\);
    \item From state \(\imath\), we select a control decision \(\gamma\) and a successor observation-state \(\imath'\) from the transition relation \(h^\mathfrak{B}_{DO}\). Since \(\imath\) is complete, such a transition always exists;
    \item From the reached observation-state \(\imath'\), we consider all transitions \((\imath',\sigma,\imath'')\) defined in \(\mathfrak{B}\). Since \(\imath'\) is complete, the transition \((\imath',\sigma,\imath'')\) is guaranteed to exist for each \(\sigma\in\Sigma_o\cap\gamma\cap\Lambda(X(\imath'))\);
    \item We continue the above selection process using depth-first search until no further new states are encountered.
\end{itemize}
The above step thus constructs an IS-based control structure from which an IS-based supervisor can be decoded. 

\begin{remark}
Note that our main purpose is to synthesize a supervisor that guarantees opacity. Thus, we are primarily interested in the solvability of Problem~\ref{prob:online-opa-enf} and do not consider the permissiveness in detail here. 
If one further seeks to synthesize a  maximally permissive supervisor, they can instead pick a locally maximal control decision in line~20; 
see, e.g.,  similar treatments in \cite{yin2015uniform}.     
\end{remark}

\begin{myexm}[Extract the Supervisor]
After expanding the solution space and pruning the incomplete states, we obtain the structure drawn in black and blue lines.
From the initial decision-state \((\{m_0\},\epsilon)\), we choose the available control decision \(\Sigma\), leading to an observation-state, from which two possible observations \(u_1\) and \(u_2\) are considered. 
Consider the decision-states \(\imath_1=(\{(2,\{2,...,7\},\{a,b,u_2\})\},u_2)\) and \(\imath_2=(\{\!(0,\{0\},\!\Sigma),\!(1,\!\{1,...,7\},\!\Sigma)\!\},\!u_2)\) from each two control decisions \(\Sigma\) and \(\{a,b,u_1\}\) are available, each leading to a different observation-state. qqIf we choose \(\Sigma\) at \(\imath_1\) and choose \(\{a,b,u_1\}\) at \(\imath_2\), the resulting control structure consists of the blue part in Figure~\ref{fig:alg}, which induces the supervisor whose controlled behavior is as discussed in Example~\ref{exm:opa-enforcing}.
\end{myexm}

\subsection{Correctness of the Algorithm}
We prove the correctness of Algorithm~\ref{alg:main} in this subsection. The soundness of the algorithm is relatively straightforward: it eliminates all unsafe states and therefore guarantees opacity. However, the completeness of the algorithm is highly nontrivial, since the solution space of the problem is unbounded (i.e., language-based supervisors), whereas our algorithm restricts attention \emph{a priori} to a finite solution space over IS-based control structures.

\begin{mylem}[\bf Soundness]\label{lem:sound}\upshape
    Given IS-based Control Structure \(\S\) returned by Algorithm~\ref{alg:main}, its decoded supervisor \(S\) constitutes a solution to Problem~\ref{prob:online-opa-enf}.
\end{mylem}
\begin{proof}
    The proof follows directly from Theorem~\ref{theo:sound} and the construction of \(\mathfrak{B}\), which ensures that all observation-states  are safe. Since \(\S\) is extracted as a subgraph of \(\mathfrak{B}\) by construction, all observation-states in \(\S\) are also safe. Therefore, the decoded supervisor of \(\S\) provides a solution to Problem~\ref{prob:online-opa-enf}.
\end{proof}

Next, we establish the completeness of the algorithm, which states that synthesizing a supervisor within the IS-based solution space is without loss of generality.

\begin{mylem}[\bf Completeness]\label{lem:complete}\upshape
    Algorithm~\ref{alg:main} will not return ``no solution exists" when a solution to Problem~\ref{prob:online-opa-enf} exists.
\end{mylem}

\begin{proof} 
We show that \(\imath^{\mathfrak{B}}_0\in \mathbb{I}^\mathfrak{B}_D\) holds in line~16 when an  opacity-enforcing supervisor (possibly non-IS-based) exists.

Suppose that there exists a language-based supervisor \(S:P_o(\mathcal{L}(G))\to\Gamma\) that solves Problem~\ref{prob:online-opa-enf}. We first construct a decision-observation structure \(B^S\) as follows.  We start by initializing a finite decision-state space 
\begin{equation}
    \mathbb{I}^S_D=\left\{\!\!\!\!\!\!\!\!
    \begin{array}{cc}
         &  \big(\{(\delta(s),\mathcal{E}_{a}(\mathbb{P}_S(s)),S(\alpha))\!:\!s\in P^{-1}_o(\alpha)\cap\\
         &\mathcal{L}(S/G)\},
         \sigma\big)\in\mathbb{I}\times\Sigma_o^\epsilon:
         \alpha\sigma\in P_o(\mathcal{L}(S/G))
    \end{array}\!\!\!\!
    \right\},\nonumber
\end{equation}
and initializing a finite observation-state space  
\begin{equation}
    \mathbb{I}^S_O=\left\{\!\!\!\!\!\!\!\!
    \begin{array}{cc}
         & \{(\delta(s),\mathcal{E}_{a}(\mathbb{P}_S(s)),S(\alpha))\!:\!s\in P^{-1}_o(\alpha)\cap\\
         &\mathcal{L}(S/G)\}\in\mathbb{I}:
         \alpha\in P_o(\mathcal{L}(S/G))
    \end{array}\!\!\!\!
    \right\}.\nonumber
\end{equation}
Then for each \(\alpha\sigma\in P_o(\mathcal{L}(S/G)\) and for decision-state  
\begin{equation}
    \imath=\left(\!\!\!\!\!\!\!\!
    \begin{array}{cc}
         &\{(\delta(s),\mathcal{E}_{a}(\mathbb{P}_S(s)),S(\alpha))\!: \\
         &s\in P^{-1}_o(\alpha)\cap\mathcal{L}(S/G)\},\sigma
    \end{array}\!\!\!\!
    \right)\in \mathbb{I}^S_D,\nonumber
\end{equation}
we define a transition \((\imath,S(\alpha\sigma),\imath')\) from decision-state \(\imath\) to observation-state \(\imath'\), where  
\begin{equation}
    \imath'=\left\{\!\!\!\!\!\!\!\!
    \begin{array}{cc}
         &(\delta(s),\mathcal{E}_{a}(\mathbb{P}_S(s)),S(\alpha\sigma))\!:\\
         &s\in P^{-1}_o(\alpha\sigma)\cap\mathcal{L}(S/G)
    \end{array}\!\!\!\!
    \right\}\in \mathbb{I}^S_O,\nonumber
\end{equation}
and we define transitions \((\imath',\sigma',\imath'')\) from observation-state \(\imath'\) to decision-state \(\imath''\) for each \(\sigma'\in \Sigma_o\cap S(\alpha\sigma)\cap\Lambda(X(\imath'))\), 
\[
\imath''=(\imath',\sigma')\in \mathbb{I}^S_D.
\]
Finally, we set \(\mathbb{I}^S_D=\mathbb{I}^S_D\cup\{(m_0,\epsilon)\}\) and we define transition \((\imath,S(\epsilon),\imath')\) for \(\imath= (\{m_0\},\epsilon)\)
and \(\imath'=\{\!(\delta(w),\mathcal{E}_{a}(\mathbb{P}_S(w)\!),\!S(\epsilon))\!:\!w\in(\Sigma_{uo}\!\cap\! S(\epsilon)\!)^*\}\) and transitions \((\imath',\sigma,\imath'')\) for each \(\sigma\in\Sigma_o\cap S(\epsilon)\cap\Lambda(X(\imath'))\) and \(\imath''=(\imath',\!\sigma\!)\). We set \((\{m_0\},\epsilon)\) as the initial state.

By the above construction, for each observation  \(\alpha=\sigma_1\sigma_2...\sigma_n\!\in\! P_o(\mathcal{L}(S/G))\), it  induces an unique path in \(B^S\) 
\begin{equation}
\imath_0\xrightarrow{S(\epsilon)}\imath'_0\xrightarrow{\sigma_1}\imath_1\xrightarrow{S(\sigma_1)}\cdots \xrightarrow{\sigma_n}\imath_n \xrightarrow{S(\sigma_1...\sigma_{n})}\imath'_n,\nonumber
\end{equation}
such that (i) \(X(\imath'_n)=\mathcal{E}^S_{o}(\alpha)\), and (ii) \(Q(\imath'_n)=\{\mathcal{E}_{a}(\mathbb{P}_S(s)):s\in P_o^{-1}(\alpha)\cap \L(S/G)\}\).
Since \(S\) solves Problem~\ref{prob:online-opa-enf}, we  have \(\mathcal{E}_{a}(\mathbb{P}_S(s))\not\subseteq X_S, \forall s\in \mathcal{L}(S/G)\) and thus each observation-state in \(B^S\) is safe. Finally, for each defined transition \(((\imath,\sigma),S(\alpha\sigma),\imath')\), we have (i) \(\imath'=\{(\delta(s),\mathcal{E}_{a}(\mathbb{P}_S(s)),S(\alpha\sigma)):s\in P^{-1}_o(\alpha\sigma)\cap\mathcal{L}(S/G)\}\), when \(\sigma\in\Sigma_a\), we can conclude that 
\begin{align}
    \imath'&=\left\{\!\!\!\!\!\!\!\!
    \begin{array}{cc}
         & (\delta(s\sigma w_1),\mathcal{E}_{a}(\mathbb{P}_S(s)(\sigma,S(\alpha\sigma))w_2),S(\alpha\sigma)): \\
         &s\!\in\! P^{-1}_o(\alpha)\!\cap\!\mathcal{L}(S/G),w_1\!\in\! (\Sigma_{uo}\!\cap\! S(\alpha\sigma))^*, \\
         &\mathbb{P}_S(s)(\sigma,S(\alpha\sigma))w_2=\mathbb{P}_S(s\sigma w_1)
    \end{array}\!\!
    \right\}\nonumber \\
    & =\mathbb{UR}_{S(\alpha\sigma)}\left(\!\!\left\{\!\!\!\!\!\!\!\!
    \begin{array}{cc}
         & (\delta(s\sigma),\mathcal{E}_{a}(\mathbb{P}_S(s)(\sigma,S(\alpha\sigma))),S(\alpha\sigma))\!: \\
         & s\!\in\! P^{-1}_o(\alpha)\!\cap\!\mathcal{L}(S/G)
    \end{array}\!\!\!\!
    \right\}\!\!\right) \nonumber \\
    & =\mathbb{UR}_{S(\alpha\sigma)}(\mathbb{NX}_{(\sigma,S(\alpha\sigma))}(\imath)), \nonumber
\end{align} 
and when \(\sigma\in\Sigma_{ua}\), we can conclude that 
\begin{align}
    \imath'&=\left\{\!\!\!\!\!\!\!\!
    \begin{array}{cc}
         & (\delta(s\sigma w_1),\mathcal{E}_{a}(\mathbb{P}_S(s)(\epsilon,S(\alpha\sigma))w_2),S(\alpha\sigma)): \\
         &s\!\in\! P^{-1}_o(\alpha)\!\cap\!\mathcal{L}(S/G),w_1\!\in\! (\Sigma_{uo}\!\cap\! S(\alpha\sigma))^*, \\
         &\mathbb{P}_S(s)(\epsilon,S(\alpha\sigma))w_2=\mathbb{P}_S(s\sigma w_1)
    \end{array}\!\!
    \right\}\nonumber \\
    & =\mathbb{UR}_{S(\alpha\sigma)}\left(\!\!\left\{\!\!\!\!\!\!\!\!
    \begin{array}{cc}
         & (\delta(s\sigma),\mathcal{E}_{a}(\mathbb{P}_S(s)(\epsilon,S(\alpha\sigma))),S(\alpha\sigma))\!: \\
         & s\!\in\! P^{-1}_o(\alpha)\!\cap\!\mathcal{L}(S/G)
    \end{array}\!\!\!\!
    \right\}\!\!\right) \nonumber \\
    & =\mathbb{UR}_{S(\alpha\sigma)}(\mathbb{NX}_{(\sigma,S(\alpha\sigma))}(\imath)); \nonumber
\end{align} 
and (ii) \(\Gamma(\imath')=S(\alpha\sigma)\).
And for each defined transition \((\imath',\sigma',\imath'')\), we have \(\imath''=(\imath',\sigma')\). For transitions \((\imath,S(\epsilon),\imath')\) and \((\imath',\sigma,\imath'')\) we can analogously get the same results. Thus, we
can conclude that \(B^S\sqsubseteq \mathfrak{B}\) according to \textbf{Step 1}.

Next, we prove that \(\mathfrak{B}\) is at least as large as \(B^S\) after calling the completeness check (\textbf{Step 2}). This conclusion directly follows the above construction, where each decision-state in \(B^S\) has at least one successor observation-state, and each observation-state in \(B^S\) has all transitions defined within the corresponding feasible observation events, and thus are complete. Since we have \(B^S\sqsubseteq \mathfrak{B}\), all states that are complete in \(B^S\) are also complete in \(\mathfrak{B}\), and thus will not be removed.

Therefore, we can finally conclude that the initial state of \(B^S\) is also the initial state of \(\mathfrak{B}\) after \textbf{Step 2}, i.e., \((\{m_0\},\epsilon)\in \mathbb{I}_D^\mathfrak{B}\). 
Hence,  Algorithm~\ref{alg:main} will not return ``no solution exists" when such a solution does exist, i.e., Algorithm~1 is also complete.
\end{proof}

By combining Lemmas~1 and~2, we can establish the correctness of Algorithm~\ref{alg:main}.
\begin{mythm}\upshape
    Algorithm~\ref{alg:main} correctly solves Problem~\ref{prob:online-opa-enf}, i.e., it is both sound and complete.
\end{mythm}

We conclude this section by discussing the complexity of the proposed supervisor synthesis algorithm. 
To synthesize an opacity-enforcing supervisor, first, we need to construct the structure \(\mathfrak{B}\) which contains at most \(2^{|X|\times2^{|X|}\times 2^{|\Sigma_c|}}\times|\Sigma_o|\) decision-states and \(2^{|X|\times2^{|X|}\times 2^{|\Sigma_c|}}\) observation-states, but this result can be further mitigated to \(2^{|X|\times2^{|X|}}\times 2^{|\Sigma_c|}\times|\Sigma_o|\) decision-states and \(2^{|X|\times2^{|X|}}\times 2^{|\Sigma_c|}\) observation-states since the information states are always consistent by construction.
For each decision-state there are at most \(2^{|\Sigma_c|}\) transitions and for each observation-state there are at most \(|\Sigma_o|\) transitions
Therefore, in the worst case, the largest possible \(\mathfrak{B}\) contains \(2^{|X|\times2^{|X|}+|\Sigma_c|}\cdot(|\Sigma_o|+1) \) states and \(2^{|X|\times2^{|X|}+|\Sigma_c|}\cdot|\Sigma_o|\cdot(2^{|\Sigma_c|}+1)\) transitions. The complexity of the completeness check procedure is quadratic in the size of \(\mathfrak{B}\). The complexity of the supervisor extraction procedure is linear in the size of \(\mathfrak{B}\).
Overall, the entire complexity of the proposed synthesis algorithm is doubly exponential in the size of the system \(G\). 
This double-exponential complexity arises due to the inherent difficulty of handling incomparable observations, as the information state involves both the intruder's state estimate and the supervisor's knowledge of the intruder. This double-exponential nature is consistent with established results for systems with incomparable observation sets \cite{dubreil2010supervisory,tong2018current,xie2021opacity}.

\section{From Observation-Triggered to Decision-triggered Decision Issuance}\label{section-change}
In the above sections, we have solved the opacity-enforcing supervisory control problem under the assumption \textbf{A4}, which states that ``The supervisor issues the control decision \emph{immediately} upon each new observable event occurrence."
In many applications, e.g., for  security or energy/bandwidth-saving purposes, the supervisor does not need to resend a new control decision to the actuator if it is the same as the currently applied one. Therefore, instead of assumption \textbf{A4} capturing the observation-triggered decision-issuance mechanism, we further investigate the \emph{decision-triggered mechanism}, which is captured by assumption \textbf{A4$^\prime$}.

\begin{itemize}
    \item [\textbf{A4$^\prime$}:] Upon each new observable event, the supervisor sends a new control decision only when it differs from the previous one; otherwise, the previous decision is maintained.
\end{itemize}
Hereafter, we show how to apply the proposed approach to solve Problem~\ref{prob:online-opa-enf} when assumption \textbf{A4}
is modified to \textbf{A4$^\prime$}.

Recalled that, 
for any string $s\in \L(S/G)$, 
the information-flow from the intruder's perspective in the observation-triggered setting is defined by 
$\mathbb{P}_S(s) \in  (\Sigma_a^\epsilon  \times \Gamma^\epsilon)^*$. 
However, in the decision-triggered setting, the intruder has less information. To account for this, the \emph{information-flow} under assumption \textbf{A4$^\prime$} is captured by a new function:
\[
\Tilde{\mathbb{P}}_S: \L(S/G)\to (\Sigma_a^\epsilon  \times \Gamma^\epsilon)^* 
\]
such that 
\begin{itemize}
    \item[(i)] $ \Tilde{\mathbb{P}}_S(\epsilon)=(\epsilon, S(\epsilon) )$; and
\item[(ii)]    for any string $s\sigma \in \L(S/G)$, we have
\begin{align} 
	\Tilde{\mathbb{P}}_S(s\sigma) \!=\!  \left\{\!\!\!
		\begin{array}{ll}
			 \Tilde{\mathbb{P}}_S(s)(\sigma ,\epsilon)            & \text{if }       
        [\sigma \in \Sigma_a]\wedge \neg\mathbf{D}   \\
			 \Tilde{\mathbb{P}}_S(s) (\epsilon, S(P_o(s\sigma)))  & \text{if }    
        [\sigma \not\in \Sigma_a]\wedge \mathbf{D}   \\
          \Tilde{\mathbb{P}}_S(s) (\sigma, S(P_o(s\sigma)))    & \text{if }    
          [\sigma \in \Sigma_a] \wedge \mathbf{D} \\
          \Tilde{\mathbb{P}}_S(s)  & \text{if }     
        [\sigma \not\in \Sigma_a]\wedge \neg\mathbf{D}
		\end{array}
		\right.\!\!\!\!\!, 
\end{align}
where  $ \mathbf{D}$ 
is the predicate that captures whether a decision change occurs, i.e., 
$ \mathbf{D}=\textsf{true}$ iff $S(P_o(s))\neq S(P_o(s\sigma))$.
\end{itemize}

Compared with the mapping $\mathbb{P}_S$, the main difference in $\Tilde{\mathbb{P}}_S$ is that we must further account for decision changes. That is, for events in $\Sigma_o$, although the supervisor always updates its decision, such a decision is visible to the intruder only when the condition $\mathbf{D}$ is satisfied. This ensures that the intruder only observes a change in the control decision when the supervisor's control decision actually differs from the previous one, reflecting the decision-triggered nature of the system.
Accordingly, the controlled state estimate under the decision-triggered mechanism is defined analogously to $ \mathcal{E}_{a}(\alpha)$ by  modifying  ${\mathbb{P}}_S$ to $\Tilde{\mathbb{P}}_S$, i.e., 
    \begin{equation}
        \Tilde{\mathcal{E}}_{a}(\alpha)\!=\!\{\delta(s)\in X: \exists S'\in \mathbb{S}\text{ s.t. }\Tilde{\mathbb{P}}_{S'}(s)=\alpha\}. 
    \end{equation}
Thus, the definition of opacity in Definition~\ref{def:online-cur-opa} should also be modified based on $\Tilde{\mathcal{E}}_{a}$ rather than ${\mathcal{E}}_{a}$. 

\begin{myexm}[Decision-triggered Online Observation]
We  still consider the running example system 
and consider  supervisor \(S:P_o(\L(G))\to \Gamma\) such that 
\(S(u_1)=S(u_1u_2)=\{a,b,u_2\}\), \(S(u_2)=S(u_2u_1)=\{a,b,u_1\}\), and \(S(\alpha)=\Sigma\) for any other \(\alpha\in P_o(\L(G))\). 
Then for string \(a u_1 u_2 u_2\in\L(S/G)\), its information-flow under observation-triggered  mechanism 
is 
\[
{\mathbb{P}}_S(s)=(\epsilon,\Sigma)(a,\epsilon)(\epsilon,\{a,b,u_2\})(\epsilon,\{a,b,u_2\})(\epsilon,\Sigma).
\]
Under the decision-triggered mechanism, its information-flow is 
\[
\Tilde{\mathbb{P}}_S(s)=(\epsilon,\Sigma)(a,\epsilon)(\epsilon,\{a,b,u_2\})(\epsilon,\Sigma).
\]
Here, the intruder does not see the second $(\epsilon,\{a,b,u_2\})$ as the control decision is not changed.
\end{myexm}
In order to synthesize a supervisor enforcing opacity for the decision-triggered mechanism, we can still follow the same approach used for the observation-triggered mechanism. Specifically, we need to:
\begin{itemize}
    \item[1)] 
    Compute the modified state estimator of the intruder based on the new information-flow;
    \item[2)] 
    Compute the knowledge of the supervisor regarding the intruder's state estimate from its own perspective;
    \item[3)] 
    Find an IS-based control structure with no bad states.
\end{itemize}
The last two parts follow exactly the same procedure as in the observation-triggered case. However, for the first part, we need to make the following modifications.

\textbf{Modified State Estimator.} 
In order to compute  $\Tilde{\mathcal{E}}_{a}$, one can modified the intruder state estimator $\mathfrak{A}$
to
\[
\Tilde{\mathfrak{A}} = 
(M,\hat{\mathcal{O}}, \Tilde{f}, m_0),
\]
where \(\Tilde{f}:M\times \hat{\mathcal{O}}\to M\) is modified from $f$ such that, 
for each state $m=(x,q,\gamma)\in M $, event  \(\sigma\in\Lambda(x)\cap\gamma\) and control decision  \(\gamma' \in \Gamma\),  we have 
   $\Tilde{f}(m, (\sigma,\gamma') )=( \delta(x,\sigma), q', \gamma'   )$, 
   where 
   \begin{equation}\label{eq:modified}
 \begin{aligned}
	q' \!=\! 
		\left\{
		\begin{array}{ll}
		    \text{UR}_{\gamma,a}(\text{NX}_\sigma(q))
            & \text{if }     \sigma \in \Sigma_a \text{ and } \gamma\!=\!\gamma'  \\
			  \text{UR}_{\gamma',a}(\text{UR}^{+}_{\gamma,a}(q))
            & \text{if }     \sigma \not\in  \Sigma_a  \text{ and } \gamma\!\neq\!\gamma'  \\
          \text{UR}_{\gamma',a}(\text{NX}_\sigma(q))
          & \text{if }     \sigma \in \Sigma_a \text{ and } \gamma\!\neq\!\gamma'\\
          q  & \text{if }     \sigma \not\in   \Sigma_{a} \text{ and } \gamma\!=\!\gamma'
		\end{array}
		\right..\!\!   
\end{aligned} 
\end{equation}
Compared with ${\mathfrak{A}} $, the main difference in the modified $\Tilde{\mathfrak{A}} $  
 is that the intruder will update its state estimate only when (i) it observes a new event, and (ii) the decision made by the supervisor is different from the previous one. Similarly, as in Proposition~\ref{prop:1}, it is easy to show that for any string
\(s\in \mathcal{L}(S/G)\),  we have 
\begin{equation}
 \Tilde{f}(\hat{s})=(\delta(s), \Tilde{\mathcal{E}}_{a}(\Tilde{\mathbb{P}}_S(s)),S(P_o(s))),    
\end{equation}
i.e., the modified 
$\Tilde{\mathfrak{A}}$   correctly computes the state estimator of the intruder under the decision-triggered mechanism. 

Based on the modified state estimator, the IS-based control structure, as well as the associated synthesis algorithm, is the same as the observation-triggered case. 
The only difference is that 
when defining operators
$\mathbb{NX}_{(\sigma,\gamma)}(\imath)$
and
$\mathbb{UR}_\gamma(\imath)$, 
one should use the modified transition function 
$\tilde{f}$ in $\tilde{\A}$ to estimate the knowledge of the intruder
rather than using the original one $f$. 

Here, we use the following example to illustrate a solution IS-based control structure for the decision-triggered case.

\begin{myexm}[Control Structure under Decision-Triggered Mechanism]

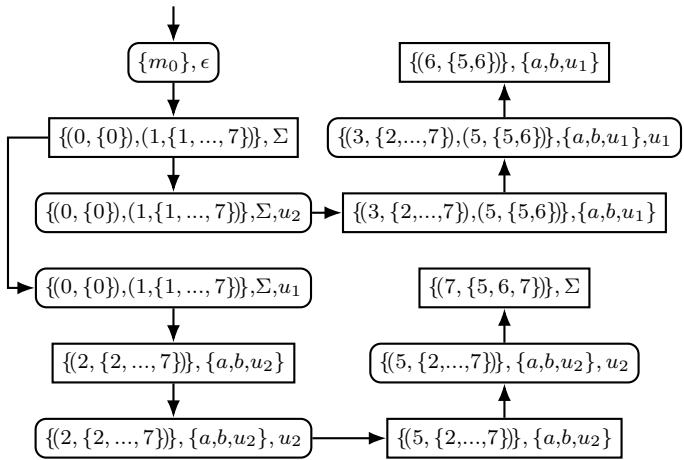
\begin{figure}
    \centering
        \begin{tikzpicture}[->,>={Latex}, thick, initial text={}, node distance= 1 cm, initial where=above, thick, D node/.style={rectangle, rounded corners, draw, minimum size=4mm, font=\footnotesize}, O node/.style={rectangle, draw, minimum size=4mm, font=\footnotesize}]    
   \node[initial, state, D node] (0) { $ \{m_0\},\epsilon $ } ;  
   \node[state, O node, ] [below of = 0] (1) {\(\{\!(0,\{0\}\!),\!(1,\!\{1,...,7\}\!)\!\}, \Sigma\)};   
   \node[state, D node][below of = 1] (2) {\(\{\!(0,\{0\}\!),\!(1,\!\{1,...,7\}\!)\!\},\! \Sigma,\!u_2\)};
   \node[state, D node][below of = 2] (3) {\(\{\!(0,\{0\}\!),\!(1,\!\{1,...,7\}\!)\!\}, \!\Sigma,\!u_1\)};
   \node[state, O node][below of = 3] (4) {\(\{\!(2,\{2,...,7\}\!)\!\}, \{a,\!b,\!u_2\}\)};
   \node[state, D node][below of = 4] (5) {\(\{\!(2,\{2,...,7\}\!)\!\}, \{a,\!b,\!u_2\},u_2\)};
   
   \node[state, O node][xshift=4.4cm] (6) {\(\{\!(6,\{5,\!6\}\!)\!\}, \{a,\!b,\!u_1\}\)};
   \node[state, D node][below of = 6] (7) {\(\{\!(3,\{2,\!...,\!7\}\!),\!(5,\{5,\!6\}\!)\!\},\! \{a,\!b,\!u_1\},\!u_1\)};
   \node[state, O node][below of = 7] (8) {\(\{\!(3,\{2,\!...,\!7\}\!),\!(5,\{5,\!6\}\!)\!\},\! \{a,\!b,\!u_1\}\)};
   \node[state, O node][below of = 8] (9) {\(\{\!(7,\{5,6,7\}\!)\!\}, \Sigma\)};
   \node[state, D node][below of = 9] (10) {\(\{\!(5,\{2,\!...,\!7\}\!)\!\}, \{a,\!b,\!u_2\},u_2\)};
   \node[state, O node][below of = 10] (11) {\(\{\!(5,\{2,\!...,\!7\}\!)\!\}, \{a,\!b,\!u_2\}\)};
   
    \draw[->] (1) -- (-2.2cm, -1cm)--(-2.2cm, -3cm)--(3);
   
   \path[->]
   (0) edge node [] {} (1)
   (1) edge node [] {} (2)
   (3) edge node [] {} (4)
   (4) edge node [] {} (5)
   (7) edge node [] {} (6)
   (8) edge node [] {} (7)
   (2) edge node [] {} (8)
   (10) edge node [] {} (9)
   (11) edge node [] {} (10)
   (5) edge node [] {} (11);
   \end{tikzpicture}
    \caption{IS-based Control Structure \(\Tilde{\S}\) of \(S\) in Example~\ref{example1}.}
    \label{fig:controlstructure2}
\end{figure}
We still consider the system shown in Figure~\ref{fig:motiexample}. The IS-based control structure \(\Tilde{\S}\) that induces the previously considered supervisor \(S\) is shown in Figure~\ref{fig:controlstructure2}, 
where  (\(S(u_1)=S(u_1u_2)=\{a,b,u_2\}\), \(S(u_2)=S(u_2u_1)=\{a,b,u_1\}\), and \(S(\alpha)=\Sigma\) for any other \(\alpha\in P_o(\L(G))\)). 
Recall that this supervisor cannot enforce current-state opacity under the assumption \textbf{A4}, as we have shown in Example~\ref{exm:opa-enforcing}.

 However, under   assumption \textbf{A4$^\prime$}, 
 this supervisor enforces opacity. 
 This is because we need to use a different transition function $\tilde{f}$ to estimate the intruder's state estimate, 
 and  the previously unsafe observation-state \(\{\!(7,\{7\},\!\Sigma)\!\}\) now becomes a safe one \(\{\!(7,\{5,6,7\},\!\Sigma)\!\}\), and the corresponding controlled state estimate now becomes
\[
\Tilde{\mathcal{E}}_a((\epsilon,\Sigma)(a,\epsilon)(\epsilon,\{a,b,u_2\})(\epsilon,\Sigma))=\{5,6,7\}.
\]
\end{myexm}

\section{Conclusion}\label{section-clu}
In this paper, we revisited the opacity-enforcing supervisor control problem from a new perspective. Our work provides a sound and complete solution to this problem under the new \emph{a priori unknown supervisor} setting. This bridges the gap between the previous fully unknown and \emph{a priori} known settings, as it preserves the solvability results of the former while allowing the intruder to leverage certain knowledge of the supervisor, as in the latter. Furthermore, our results offer a complete solution to the standard opacity control problem when the intruder has finer information, even when the supervisor is fully known \emph{a priori}.
In future work, we would like to extend our framework to the stochastic setting by quantifying the level of opacity \cite{lefebvre2020exposure,lefebvre2020privacy,udupa2025synthesis}, as well as to continuous dynamic systems by considering uncountable state spaces \cite{liu2020verification,zhong2025secure}.%

\bibliographystyle{plain}
\bibliography{references}

\appendix
\textbf{Proof of Proposition~\ref{prop:1}} 
\begin{proof}
    The conclusions for the first and third components are straightforward according to the definition of the intruder state estimator. We now prove the second component by induction on the length of \(\hat{s}\):
    
    \emph{Induction Basis}: We first suppose that \(\hat{s}=(\epsilon,S(\epsilon))\). By Definition~\ref{def:on-curr-est}, we have \(\mathcal{E}_{a}(\mathbb{P}_S(s))\!=\!\{\delta(s)\in X: \exists S'\in\mathbb{S} \text{ s.t. } \mathbb{P}_{S'}(s)=(\epsilon,S(\epsilon))\}=\{\delta(s)\in X: \exists S'\in\mathbb{S} \text{ s.t. } S'(\epsilon)=S(\epsilon), s\in \L(S'/G), P_a(s)=\epsilon\}=\{\delta(s)\in X: s\in (\Sigma_{ua}\cap S(\epsilon))^*\}=\text{UR}_{S(\epsilon),a}(\{x_0\})\), which aligns with Equation~\eqref{eq:9}. Thus, the induction basis holds.

    \emph{Induction Step}: 
    Now we suppose that the second component equals to \(\mathcal{E}_{a}(\mathbb{P}_S(s))\) holds for \(|\hat{s}|=k\). For the augmented string \(\hat{s}(\sigma,\gamma')\), we first note that for all \(S'\in\mathbb{S}\) such that there exists a string \(t\in\mathcal{L}(S'/G): \mathbb{P}_{S'}(t)=\mathbb{P}_S(s)(\sigma,\gamma')\), there must exists a string \(t'\in\overline{\{t\}}:\mathbb{P}_{S'}(t')=\mathbb{P}_S(s)\), we then accordingly consider the following three cases:
    \begin{itemize}
        \item [(1)] If \(\sigma\in\Sigma_a\setminus\Sigma_o\), we have \vspace{-3pt}
        \begin{align}
                 &\mathcal{E}_{a}(\mathbb{P}_S(s)(\sigma,\epsilon))\nonumber\\
                 &=\{\delta(t)\in X: \exists S'\in\mathbb{S}\text{ s.t. } \mathbb{P}_{S'}(t)=\mathbb{P}_S(s)(\sigma,\epsilon)\}\nonumber\\
                 &=\left\{
                    \!\!\!\!\!\!\begin{array}{rr}
                         &  \delta(\delta(t'),w)\in X: \delta(t')\in \mathcal{E}_{a}(\mathbb{P}_S(s)), \\
                         & \exists S'\in\mathbb{S}: \mathbb{P}_{S'}(t'w)=\mathbb{P}_S(s)(\sigma,\epsilon) 
                    \end{array}\right\} \nonumber\\
                 &= \left\{
                    \!\!\!\!\!\!\begin{array}{rr}
                         &  \delta(\delta(t'),w)\in X: \delta(t')\in \mathcal{E}_{a}(\mathbb{P}_S(s)), \\
                         & w\in\{\sigma\}\cdot(\Sigma_{ua}\cap\gamma)^*
                    \end{array}\right\} \nonumber\\
                &= \text{UR}_{\gamma,a}(\!\{\delta(\delta(t'),\sigma)\!\in\! X\!:\! \delta(t')\!\in\! \mathcal{E}_{a}(\mathbb{P}_S(s))\}\!)\nonumber\\
                &= \text{UR}_{\gamma,a}(\text{NX}_{\sigma}(\mathcal{E}_{a}(\mathbb{P}_S(s))))\nonumber
        \end{align}
        \item [(2)] If \(\sigma\in\Sigma_o\setminus\Sigma_a\), we have \vspace{-3pt}
        \begin{align}
                 & \mathcal{E}_{a}(\mathbb{P}_S(s)(\epsilon,\gamma'))\nonumber\\
                 & =\{\delta(t)\in X: \exists S'\in\mathbb{S}\text{ s.t. } \mathbb{P}_{S'}(t)=\mathbb{P}_S(s)(\epsilon,\gamma')\}\nonumber\\
                 &=\left\{
                    \!\!\!\!\!\!\begin{array}{rr}
                         &  \delta(\delta(t'),w)\in X: \delta(t')\in \mathcal{E}_{a}(\mathbb{P}_S(s)), \\
                         & \exists S'\in\mathbb{S}: \mathbb{P}_{S'}(t'w)=\mathbb{P}_S(s)(\epsilon,\gamma')
                    \end{array}\right\}\nonumber\\
                &=\left\{
                    \!\!\!\!\!\!\begin{array}{rr}
                         &  \delta(\delta(t'),w)\in X: \delta(t')\in \mathcal{E}_{a}(\mathbb{P}_S(s)), \\
                         & w\in((\Sigma_{ua}\cap\gamma)^*\setminus\{\epsilon\})\cdot(\Sigma_{ua}\cap\gamma')^*
                    \end{array}\right\} \nonumber\\
                &= \!\!\!\!\!\!\begin{array}{rr}
                     & \text{UR}_{\gamma',a}(\{ \delta(\delta(t'),w')\in X: \delta(t')\in \mathcal{E}_{a}(\mathbb{P}_S(s))\\
                     & w'\in ((\Sigma_{ua}\cap\gamma)^*\setminus\{\epsilon\})
                \end{array} \nonumber\\
                &=\text{UR}_{\gamma',a}(\text{UR}^{+}_{\gamma,a}(\mathcal{E}_{a}(\mathbb{P}_S(s))))\nonumber
        \end{align}
        \item [(3)]If \(\sigma\in\Sigma_o\cap\Sigma_a\), we have \vspace{-3pt}
        \begin{align}
                 & \mathcal{E}_{a}(\mathbb{P}_S(s)(\sigma,\gamma'))\nonumber\\
                 & =\{\delta(t)\in X: \exists S'\in\mathbb{S}\text{ s.t. } \mathbb{P}_{S'}(t)=\mathbb{P}_S(s)(\sigma,\gamma')\}\nonumber\\
                 &=\left\{
                    \!\!\!\!\!\!\begin{array}{rr}
                         &  \delta(\delta(t'),w)\in X: \delta(t')\in \mathcal{E}_{a}(\mathbb{P}_S(s)), \\
                         & \exists S'\in\mathbb{S}: \mathbb{P}_{S'}(t'w)=\mathbb{P}_S(s)(\sigma,\gamma')
                    \end{array}\right\} \nonumber\\
                &=\left\{
                    \!\!\!\!\!\!\begin{array}{rr}
                         &  \delta(\delta(t'),w)\in X: \delta(t')\in \mathcal{E}_{a}(\mathbb{P}_S(s)), \\
                         & w\in\{\sigma\}\cdot(\Sigma_{ua}\cap\gamma')^*
                    \end{array}\right\} \nonumber\\
                &=\text{UR}_{\gamma',a}(\!\{\delta(\delta(t'),\sigma)\!\in\! X\!:\! \delta(t')\!\in \!\mathcal{E}_{a}(\mathbb{P}_S(s))\}\!)\nonumber\\
                &=\text{UR}_{\gamma',a}(\text{NX}_{\sigma}(\mathcal{E}_{a}(\mathbb{P}_S(s))))\nonumber
        \end{align}

    \end{itemize}
    The above three conditions all align with Equation~\eqref{eq:10}, and the proof is thus completed.
\end{proof}

\textbf{Proof of Proposition~\ref{prop:IS}} 
\begin{proof}
    We prove by induction on the length of $\alpha$. 

    \emph{Induction Basis}:
    Suppose that $|\alpha|=0$, i.e., \(\alpha=\epsilon\). 
    We have \vspace{-3pt}
    \begin{align}
        \mathbb{I}^\S_O(\epsilon)&=h^\S_{DO}(\imath_0^\S,S(\epsilon)) \nonumber\\
        &=\mathbb{UR}_{S(\epsilon)}({\mathbb{NX}}_{(\epsilon,S(\epsilon))}(\{m_0\}))\nonumber\\
        &={\mathbb{UR}}_{S(\epsilon)}(\{(x_0,\text{UR}_{S(\epsilon),a}(\{x_0\}),S(\epsilon))\})\nonumber
    \end{align}
    According to Proposition~\ref{prop:1}, \(\text{UR}_{S(\epsilon),a}(\{x_0\})=\mathcal{E}_{a}((\epsilon,S(\epsilon)))\), thus we can further conclude \(\mathbb{I}^\S_O(\epsilon)=\{(x,q,\gamma)\in M : (x,q,\gamma)=f((x,\mathcal{E}_{a}((\epsilon,S(\epsilon))),S(\epsilon)),w), w\in ((\Sigma_{uo}\cap S(\epsilon))\times\{S(\epsilon)\})^*\}=\{(x,q,\gamma)\in M : (x,q,\gamma)=f((\epsilon,S(\epsilon))w), w\in ((\Sigma_{uo}\cap S(\epsilon))\times\{S(\epsilon)\})^*\}\). Then $X(\mathbb{I}^\S_O(\epsilon))=\{ \delta(w) \in X: w \in (\Sigma_{uo}\cap S(\epsilon))^* \}=\mathcal{E}^S_{o}(\epsilon)$ and
    \(Q(\mathbb{I}^\S_O(\epsilon))=\{\mathcal{E}_{a}((\epsilon,S(\epsilon))w): w\in((\Sigma_{uo}\cap\Sigma_a\cap S(\epsilon))\times\{\epsilon\})^*\}=\{\mathcal{E}_{a}(\mathbb{P}_S(s)):s\in P_o^{-1}(\epsilon)\cap \L(S/G)\}\). The induction basis thus holds.

    \emph{Induction Step}:
    Now we suppose that the conclusion holds for \(|\alpha|=k\), we prove that it also holds for \(\alpha\sigma\in P_o(\L(S/G))\), where \(\sigma\in\Sigma_o\). According to Def.~\ref{def:IS-Cstructure}, we have that \vspace{-3pt}
    \[
    \mathbb{I}^\S_O(\alpha\sigma) = {\mathbb{UR}}_{S(\alpha\sigma)}({\mathbb{NX}}_{(\sigma,S(\alpha\sigma))}(\mathbb{I}^\S_O(\alpha))
    \]   
    When \(\sigma\in\Sigma_o\cap\Sigma_a\), we have \vspace{-3pt}
    \begin{align}
        &\mathbb{I}^\S_O(\alpha\sigma) \nonumber\\
        & ={\mathbb{UR}}_{S(\alpha\sigma)}
        \left(\!\!
            \left\{\!\!\!\!\!\!\!\!
                \begin{array}{ll}
                        & (x',q',\gamma')\!\in\! M: \exists (x,q,\gamma)\!\in\! \mathbb{I}^\S_O(\alpha) \\ &(x',\!q',\!\gamma'\!)\!=\! f(\!(x,\!q,\!\gamma),\!(\sigma,\! S(\alpha\sigma)\!)\!)
                \end{array}\!\!\!\!
            \right\}\!\!
        \right) \nonumber \\
        & ={\mathbb{UR}}_{S(\alpha\sigma)}
        \left(\!\!
            \left\{\!\!\!\!\!\!\!\!
                \begin{array}{ll}
                        & (\delta(s\sigma),\mathcal{E}_{a}(\mathbb{P}_S(s)(\sigma,S(\alpha\sigma))),S(\alpha\sigma)):\\
                        & s\in P^{-1}_o(\alpha)\cap\L(S/G)
                \end{array}\!\!
            \right\}\!\!
        \right) \nonumber\\
        & = \left\{\!\!\!\!\!\!\!\!
        \begin{array}{ll}
             &  (x,q,\gamma)\in M: (x,q,\gamma)=\\
             & f((\delta(s\sigma),
             \mathcal{E}_{a}(\mathbb{P}_S(s)(\sigma,S(\alpha\sigma))),S(\alpha\sigma)),w)\nonumber \\
             & s\in P^{-1}_o(\alpha)\!\cap\!\L(S/G), w\in (\!(\Sigma_{uo}\!\cap\! S(\alpha\sigma)\!)\!\times\!\{S(\alpha\sigma)\}\!)^*
        \end{array}\!\! 
        \right\}\nonumber      
    \end{align}
    Then we can conclude that \vspace{-3pt}
    \begin{align}
        X(\mathbb{I}^\S_O(\alpha\sigma)) &=\left\{\delta(s\sigma w)\in X: \!\!\!\!
        \begin{array}{cc}
             &s\in P^{-1}_o(\alpha)\cap\L(S/G), \\
             & w\in(\Sigma_{uo}\cap S(\alpha\sigma))^*
        \end{array}\right\}\nonumber \\
        &=\{\delta(s)\in X: s\in P^{-1}_o(\alpha\sigma)\cap\L(S/G)\}\nonumber \\
        &=\mathcal{E}^S_{o}(\alpha\sigma)\nonumber
    \end{align}
    and \vspace{-3pt}
    \begin{align}
    Q(\mathbb{I}^\S_O(\alpha\sigma))
        & = \left\{\!\!\!\!\!\!\!\!
                \begin{array}{ll}
                        & \mathcal{E}_{a}(\mathbb{P}_S(s)(\sigma,S(\alpha\sigma))w):\\
                        & s\in P^{-1}_o(\alpha)\cap\L(S/G),\\ & w\in((\Sigma_{uo}\!\cap\!\Sigma_a\!\cap\! S(\alpha\sigma))\!\times\!\{\epsilon\})^*
                \end{array}\!\!
        \right\} \nonumber \\
        & = \{\mathcal{E}_{a}(\mathbb{P}_S(s)):s\in P_o^{-1}(\alpha\sigma)\cap\L(S/G)\}\nonumber
    \end{align}
    When \(\sigma\in\Sigma_o\cap\Sigma_{ua}\), we have \vspace{-3pt}
    \begin{align}
        &\mathbb{I}^\S_O(\alpha\sigma) \nonumber\\
        & ={\mathbb{UR}}_{S(\alpha\sigma)}
        \left(\!\!
            \left\{\!\!\!\!\!\!\!\!
                \begin{array}{ll}
                        & (x',q',\gamma')\!\in\! M: \exists (x,q,\gamma)\!\in\! \mathbb{I}^\S_O(\alpha) \\ &(x',\!q',\!\gamma'\!)\!=\! f(\!(x,\!q,\!\gamma),\!(\sigma,\! S(\alpha\sigma)\!)\!)
                \end{array}\!\!\!\!
            \right\}\!\!
        \right) \nonumber \\
        & ={\mathbb{UR}}_{S(\alpha\sigma)}
        \left(\!\!
            \left\{\!\!\!\!\!\!\!\!
                \begin{array}{ll}
                        & (\delta(s\sigma),\mathcal{E}_{a}(\mathbb{P}_S(s)(\epsilon,S(\alpha\sigma))),S(\alpha\sigma)):\\
                        & s\in P^{-1}_o(\alpha)\cap\L(S/G)
                \end{array}\!\!
            \right\}\!\!
        \right) \nonumber\\
        & = \left\{\!\!\!\!\!\!\!\!
        \begin{array}{ll}
             &  (x,q,\gamma)\in M: (x,q,\gamma)=\\
             & f((\delta(s\sigma),
             \mathcal{E}_{a}(\mathbb{P}_S(s)(\epsilon,S(\alpha\sigma))),S(\alpha\sigma)),w)\nonumber \\
             & s\in P^{-1}_o(\alpha)\!\cap\!\L(S/G), w\in (\!(\Sigma_{uo}\!\cap\! S(\alpha\sigma)\!)\!\times\!\{S(\alpha\sigma)\}\!)^*
        \end{array}\!\! 
        \right\}\nonumber      
    \end{align}
    We can also conclude that \vspace{-3pt}
    \begin{align}
        X(\mathbb{I}^\S_O(\alpha\sigma))=\mathcal{E}^S_{o}(\alpha\sigma)\nonumber
    \end{align}
    and \vspace{-3pt}
    \begin{align}
    Q(\mathbb{I}^\S_O(\alpha\sigma)) 
        & = \left\{\!\!\!\!\!\!\!\!
                \begin{array}{ll}
                        & \mathcal{E}_{a}(\mathbb{P}_S(s)(\epsilon,S(\alpha\sigma))w):\\
                        & s\in P^{-1}_o(\alpha)\cap\L(S/G),\\ &w\in((\Sigma_{uo}\!\cap\!\Sigma_a\!\cap\! S(\alpha\sigma))\!\times\!\{\epsilon\})^*
                \end{array}\!\!
        \right\} \nonumber \\
        & = \{\mathcal{E}_{a}(\mathbb{P}_S(s)):s\in P_o^{-1}(\alpha\sigma)\cap\L(S/G)\}\nonumber
    \end{align}
    The induction step is thus completed.
\end{proof}

\end{document}